\documentclass[11pt]{article}

\usepackage{amssymb, amsfonts,amsthm,fullpage}
\usepackage{float}
\usepackage{comment}
\usepackage{color}
\usepackage{graphicx}
\usepackage{mathtools}
\usepackage{bbm}
\usepackage{algorithm,algorithmic}
\usepackage{enumitem}
\usepackage{comment}
\usepackage{xifthen}
\usepackage{breakcites}
\usepackage{xcolor}

\usepackage{parskip}

\usepackage{nicefrac}

\usepackage[labeled]{multibib}
\newcites{App}{Additional References}

\usepackage[colorlinks=true,linkcolor=blue,backref=page]{hyperref}
\usepackage[capitalise,noabbrev]{cleveref}
\usepackage[all]{hypcap}
\usepackage{bm}

\newtheorem{theorem}{Theorem}[section]%
\newtheorem{lemma}[theorem]{Lemma}
\newtheorem{definition}[theorem]{Definition}
\newtheorem{corollary}[theorem]{Corollary}

\newtheorem{question}[theorem]{Question}
\newtheorem{remark}[theorem]{Remark}

\newtheorem{proposition}{Proposition}

\DeclareMathOperator{\E}{E}

\newcommand{\AAA}{\mathcal A}
\newcommand{\AAAspace}{{\mathcal A}_{\rm space}}
\newcommand{\BBB}{\mathcal B}
\newcommand{\BBBsamp}{{\mathcal B}_{\rm sample}}
\newcommand{\BBBspace}{{\mathcal B}_{\rm space}}
\newcommand{\hBBBspace}{\widehat{\mathcal B}_{\rm space}}
\newcommand{\WWW}{\mathcal W}
\newcommand{\oracle}{\mathcal{E}}

\newcommand{\DDD}{\mathcal D}
\newcommand{\DDDsamp}{{\mathcal D}_{\rm sample}}
\newcommand{\DDDspace}{{\mathcal D}_{\rm space}}

\newcommand{\PPP}{\mathcal P}

\newcommand{\XXX}{\mathcal X}
\newcommand{\YYY}{\mathcal Y}
\newcommand{\eps}{\varepsilon}

\newcommand{\Gen}{{\rm Gen}}
\newcommand{\Param}{{\rm Param}}
\newcommand{\Enc}{{\rm Enc}}

\newcommand{\Dec}{{\rm Dec}}

\newcommand{\MAJ}{\operatorname{\rm MAJ}}

\newcommand{\polylog}{\mathop{\rm polylog}}

\newcommand{\poly}{\mathop{\rm poly}}

\newcommand{\negl}{\mathop{\rm negl}}

\def\E{\operatorname*{\mathbb{E}}}

\def\poly{\mathop{\rm{poly}}\nolimits}

\newcommand{\mynote}[2]{{\textcolor{#1}{ #2}}}

\newcommand{\red}[1]{\mynote{red}{#1}}
\newcommand{\gray}[1]{\mynote{gray}{#1}}

\DeclareSymbolFont{AMSb}{U}{msb}{m}{n}
\DeclareMathSymbol{\N}{\mathbin}{AMSb}{"4E}
\DeclareMathSymbol{\Z}{\mathbin}{AMSb}{"5A}
\DeclareMathSymbol{\R}{\mathbin}{AMSb}{"52}
\DeclareMathSymbol{\erert}{\mathbin}{AMSb}{"50}
\DeclareMathSymbol{\I}{\mathbin}{AMSb}{"49}
\DeclareMathSymbol{\C}{\mathbin}{AMSb}{"43}

\newcommand{\adv}{{\rm Adv}}

\author{
Itai Dinur\thanks{Ben-Gurion University. \texttt{dinuri@bgu.ac.il}.
Partially supported by the Israel Science Foundation (grant 1903/20) and by the European Research Council under the ERC starting grant agreement no. 757731 (LightCrypt).}
\and
Uri Stemmer\thanks{Tel Aviv University and Google Research. \texttt{u@uri.co.il}. Partially supported by the Israel Science Foundation (grant 1871/19) and by
Len Blavatnik and the Blavatnik Family foundation.}
\and
David P. Woodruff\thanks{Carnegie Mellon University. \texttt{dwoodruf@andrew.cmu.edu}. Work done in part while visiting Google Research. Also supported by a Simons Investigator Award and by the National Science Foundation under Grant No. CCF-1815840.}
\and
Samson Zhou\thanks{UC Berkeley and Rice University. \texttt{samsonzhou@gmail.com}. Work done in part while at Carnegie Mellon University. Partially supported by a Simons Investigator Award and by the National Science Foundation under Grant No. CCF-1815840.}
}

\title{On Differential Privacy and Adaptive Data Analysis with Bounded Space}

\date{February 11, 2023}

\begin{document}

\maketitle

\begin{abstract}
We study the space complexity of the two related fields of {\em differential privacy} and {\em adaptive data analysis}. Specifically,

\begin{enumerate}[topsep=5px, itemsep=5
px]
    \item Under standard cryptographic assumptions, we show that there exists a problem $P$ that requires exponentially more space to be solved efficiently with differential privacy, compared to the space needed without privacy. To the best of our knowledge, this is the first separation between the space complexity of private and non-private algorithms.
    
    \item The line of work on adaptive data analysis focuses on understanding the number of {\em samples} needed for answering a sequence of adaptive queries. We revisit previous lower bounds at a foundational level, and show that they are a consequence of a space bottleneck rather than a sampling bottleneck.
\end{enumerate}

To obtain our results, we define and construct an encryption scheme with multiple keys that is built to withstand a limited amount of key leakage in a very particular way.
\end{abstract}

\section{Introduction}

Query-to-communication lifting theorems allow translating lower bounds on the {\em query complexity}  of a given function $f$ to lower bounds on the {\em communication complexity} of a related function $\hat{f}$. Starting from the seminal work of Raz and McKenzie \cite{RazM99}, several such lifting theorems were presented, and applied, to obtain new communication complexity lower bounds in various settings. 

In the domain of cryptography, related results have been obtained, where the starting point is a lower bound on the {\em query complexity} of an adversary solving a cryptanalytic problem in an idealized model, such as the random oracle model~\cite{BellareR93}. The query complexity lower bound is then lifted to a {\em query-space} lower bound for a non-uniform (preprocessing) adversary solving the same problem~\cite{Unruh07,CorettiDGS18,CorettiDG18}. 

Building on ideas developed in these lines of work, we present a new technique for translating sampling lower bounds to space lower bounds for problems in the context of {\em differential privacy} and {\em adaptive data analysis}. Before presenting our results, we motivate our settings.

\subsection{Differential privacy}
Differential privacy~\cite{dwork2016calibrating} is a mathematical definition for privacy that aims to enable statistical analyses of datasets while providing strong guarantees that individual-level information does not leak.
Informally, an algorithm that analyzes data satisfies differential privacy if it is robust in the sense that its
outcome distribution does not depend ``too much'' on any single data point. Formally,
\begin{definition}[\cite{dwork2016calibrating}]
Let $\AAA:X^*\rightarrow Y$ be a randomized algorithm whose input is a dataset $D\in X^*$. Algorithm $\AAA$ is {\em $(\eps,\delta)$-differentially private (DP)} if for any two datasets $D,D'$ that differ on one point (such datasets are called {\em neighboring}) and for any outcome set $F\subseteq Y$ it holds that
$
\Pr[\AAA(D)\in F]\leq e^{\eps}\cdot\Pr[\AAA(D')\in F]+\delta.
$
\end{definition}
To interpret the definition, let $D$ be a dataset containing $n$ data points, each of which represents the information of one individual. Suppose that Alice knows all but one of these data points (say Bob's data point). Now suppose that we compute $z\leftarrow\AAA(D)$, and give $z$ to Alice. If $\AAA$ is differentially private, then Alice learns very little about Bob's data point, because $z$ would have been distributed roughly the same no matter what Bob's data point is.

Over the last few years, we have witnessed an explosion of research on differential privacy in various settings. In particular, a fruitful line of work has focused on designing differentially private algorithms with small {\em space} complexity, mainly in streaming settings. Works in this vein include~\cite{DworkNPRY10,mir2011pan,BlockiBDS12,BassilyNST20,BunNS19,AlabiBC22,KaplanS21,Smith0T20,PaghS21,PaghT22,bu2021fast,wang2021differentially,BolotFMNT13,chan2012differentially,Upadhyay19,UpadhyayU21,zhao2022differentially}. All of these works showed {\em positive} results and presented differentially private algorithms with small space complexity for various problems. In fact, some of these works showed that classical streaming algorithms are differentially private essentially {\em as is}. For example, Blocki et al.~\cite{BlockiBDS12} showed that the Johnson-Lindenstrauss transform itself preserves differential privacy, and Smith et al.~\cite{Smith0T20} showed this for the classical Flajolet-Martin Sketch.

In light of all these positive results, one might think that algorithms with small space are particularly suitable for differential privacy, because these algorithms are not keeping too much information about the input to begin with. 

\begin{question}\label{quest:1}
Does differential privacy require more space?
\end{question}

\subsubsection{Our results for differential privacy with
bounded space}
We answer Question~\ref{quest:1} in the affirmative, i.e., we show that differential privacy \emph{may require more space}. To this end, we come up with a problem that can be solved using a small amount of space without privacy, but requires a large amount of space to be solved with privacy. As a first step, let us examine the following toy problem, which provides {\em some} answer to the above question.

\paragraph{A non-interesting toy problem.} 
Recall that $F_2$ (the second frequency moment of a stream) estimation with multiplicative approximation error $1+\alpha$ has an $\Omega(1/\alpha^2)$ space lower bound \cite{w04}. This immediately shows a separation for the problem of ``output either the last element of the stream or a $(1+\alpha)$-approximation to the $F_2$ value of the stream''. In the non-private setting, the last element can be output using space independent of $\alpha$, but in the private setting the algorithm is forced to (privately) estimate $F_2$ and thus use at least $1/\alpha^2$ space. Of course, we could replace $F_2$ with other tasks that have a large space lower bound in the standard non-private model.

We deem this toy problem non-interesting because, at a high level, our goal is to show that there are cases where computing something privately requires a lot more space than computing ``the same thing'' non-privately. In the toy problem, however, the private and non-private algorithms are arguably {\em not} computing ``the same thing''. To reconcile this issue, we will focus on problems that are defined by a {\em function} (ranging over some metric space), and the desired task would be to approximate the value of this function. Note that this formulation disqualifies the toy problem from being a valid answer to Question~\ref{quest:1}, and that with this formulation there is a formal sense in which every algorithm for solving the task must compute (or approximate) ``the same thing''.

Let us make our setting more precise. In order to simplify the presentation, instead of studying the streaming model, we focus on the following computation model.\footnote{We remark, however, that all of our results extend to the streaming setting. See  Remark~\ref{rem:streamingModel}.} Consider an algorithm that is instantiated on a dataset $D$ and then aims to answer a query with respect to $D$. We say that such an algorithm has space $s$ if, before it gets the query, it shrinks $D$ to a {\em summary} $z$ containing at most $s$ bits. Then, when obtaining the query, the algorithm answers it using only the summary $z$ (without additional access to the original dataset $D$). 
Formally, we consider problems that are defined by a function $P:X^*\times Q\rightarrow M$, where $X$ is the data domain, $Q$ is a family of possible queries, and $M$ is a metric space.

\begin{definition}
We say that $\AAA=(\AAA_1,\AAA_2)$ solves a problem $P:X^*\times Q\rightarrow M$ with space complexity $s$, sample complexity $n$, error $\alpha$, and confidence $\beta$ if 
\begin{enumerate}
    \item $\AAA_1:X^*\rightarrow\{0,1\}^s$ is a {\em preprocessing procedure} that takes a dataset $D$ and outputs an $s$-bit string. 
    \item For every input dataset $D\in X^{n}$ and every query $q\in Q$ it holds that
    $$
    \Pr_{\substack{z\leftarrow\AAA_1(D)\\
    a\leftarrow\AAA_2(z,q)
    }}[|a-P(D,q)|\leq\alpha]\geq1-\beta.
    $$
\end{enumerate}
\end{definition}

We show the following theorem.

\begin{theorem}[informal]
Let $d\in\N$ be a parameter controlling the size of the problem (specifically, data points from $X$ can be represented using $\polylog(d)$ bits, and queries from $Q$ can be represented using $\poly(d)$ bits). There exists a problem $P:X^*\times Q\rightarrow M$ such that the following holds.
\begin{enumerate}
    \item $P$ can be solved non-privately using $\polylog(d)$ bits of space.\\ {\small \gray{\% See Lemma~\ref{lem:DAisEasyNonPrivately} for the formal statement.}}
    \item $P$ can be solved privately using sample and space complexity $\tilde{O}(\sqrt{d})$.\\ {\small \gray{\% See Lemma~\ref{lem:DAisSolvablePrivately} for the formal statement.}}
    \item Assuming the existence of a sub-exponentially secure symmetric-key encryption scheme, every computationally-efficient differentially-private algorithm $\AAA$ for solving $P$ must have space complexity $\tilde{\Omega}(\sqrt{d})$, even if its sample complexity is a large polynomial in $d$. Furthermore, this holds even if $\AAA$ is only required to be {\em computationally} differentially private (namely, the adversary we build against $\AAA$ is computationally efficient).\\ {\small \gray{\% See Corollary~\ref{cor:finalDP} for the formal statement.}}
\end{enumerate}
\end{theorem}

Note that this is an exponential separation (in $d$) between the non-private space complexity and the private space complexity. We emphasize that the hardness of privately solving $P$ does not come from not having enough samples. Indeed, by Item~2, $\tilde{O}(\sqrt{d})$ samples suffice for privately solving this problem. However, Item~3 states that unless the algorithm has large space, then it cannot privately solve this problem {\em even if it has many more samples than needed}.

To the best of our knowledge, this is the first result that separates the space complexity of private and non-private algorithms. Admittedly, the problem $P$ we define to prove the above theorem is somewhat unnatural. In contrast, our negative results for adaptive data analysis (to be surveyed next) are for the canonical problem studied in the literature. 

\subsection{Adaptive data analysis}

Consider a data analyst interested in testing a specific research hypothesis. The analyst acquires relevant data, evaluates the hypothesis, and (say) learns that it is false. Based on the findings, the analyst now decides on a second hypothesis to be tested, and evaluates it {\em on the same data} (acquiring fresh data might be too expensive or even impossible). That is, the analyst chooses the hypotheses {\em adaptively}, where this choice depends on previous interactions with the data. As a result, the findings are no longer supported by classical statistical theory, which assumes that the tested hypotheses are fixed before the data is gathered, and the analyst runs the risk of overfitting to the data.

Starting with~\cite{dwork2015preserving}, 
the line of work on {\em adaptive data analysis (ADA)} aims to design methods for provably guaranteeing statistical validity in such settings. Specifically, the goal is to design a mechanism $\AAA$ that initially obtains a dataset $D$ containing $t$ i.i.d.\ samples from some unknown distribution $\PPP$, and then answers $k$ {\em adaptively chosen queries} w.r.t.\ $\PPP$. Importantly, $\AAA$'s answers must be accurate w.r.t.\ the underlying distribution $\PPP$, and not just w.r.t.\ the empirical dataset $D$. The main question here is,
\begin{question}\label{quest:2}
How many samples does $\AAA$ need (i.e., what should $t$ be) in order to support $k$ such adaptive queries?
\end{question}

As a way of dealing with worst-case analysts, the analyst is assumed to be adversarial in that it tries to cause
the mechanism to fail. If a mechanism can maintain utility against such an adversarial analyst, then
it maintains utility against any analyst. Formally, the canonical problem pursued by the line of work on ADA is defined as a two-player game between a mechanism $\AAA$ and an adversary $\BBB$. See Figure~\ref{fig:ADAintro}.

\begin{figure}
    \centering
    \fbox{%
  \parbox{\textwidth}{%

\begin{enumerate}
    \item The adversary $\BBB$ chooses a distribution $\PPP$ over a data domain $X$.
    
    \item The mechanism $\AAA$ obtains a sample $S\sim\PPP^t$ containing $t$ i.i.d.\ samples from $\PPP$.
    
    \item For $k$ rounds $j = 1, 2,\dots, k$:
    
    \begin{itemize}[leftmargin=7px]
        \item The adversary chooses a function $h_j : X \rightarrow \{-1,0, 1\}$, possibly as a function of all previous answers given by the mechanism.
        \item The mechanism obtains $h_j$ and responds with an answer $z_j$, which is given to $\BBB$.
    \end{itemize}
\end{enumerate} 
  }
}
    \caption{A two-player game between a mechanism $\AAA$ and an adversary $\BBB$.}
    \label{fig:ADAintro}
\end{figure}

\begin{definition}[\cite{dwork2015preserving}]\label{def:ADA}
A mechanism $\AAA$ is $(\alpha, \beta)$-accurate for $k$ queries over a domain $X$ using sample size $t$ if for every adversary $\BBB$ (interacting with $\AAA$ in the game specified in Figure~\ref{fig:ADAintro}) it holds that 
$$
\Pr\big[\exists j\in[k] \text{ s.t.\ } |z_j-h_j(\PPP)|>\alpha\big]\leq\beta,
$$
where $h_j (\PPP) = \E_{x\sim\PPP} [h_j (x)]$ is the ``true'' value of $h_j$ on the underlying
distribution $\PPP$.
\end{definition}

Following Dwork et al.~\cite{dwork2015preserving}, this problem has attracted a significant amount of work~\cite{HardtU14,SteinkeU15,BassilyNSSSU21,ullman2018limits,KontorovichSS22,0001LN0SS21,abs-2106-10761,ShenfeldL19,FishRR18,FishRR20,DworkFHPRR15,rogers2016max,cummings2016adaptive,feldman2018calibrating,feldman2017generalization,bassily2016typical}, most of which is focused on understanding how many {\em samples} are needed for adaptive data analysis (i.e., focused on Question~\ref{quest:2}). In particular, the following almost matching bounds are known.

\begin{theorem}[\cite{dwork2015preserving,BassilyNSSSU21}] 
There exists a computationally efficient mechanism that is $(0.1,0.1)$-accurate for $k$ queries using sample size $t=\tilde{O}\left(\sqrt{k}\right)$.
\end{theorem}

\begin{theorem}[\cite{HardtU14,SteinkeU15}, informal]\label{thm:HSU1415}
Assuming the existence of one-way functions, every computationally efficient mechanism that is $(0.1,0.1)$-accurate for $k$ queries must have sample size at least $t=\Omega\left(\sqrt{k}\right)$.
\end{theorem}

\subsubsection{Our results for adaptive data analysis with bounded space}

All prior work on the ADA problem treated it as a {\em sampling} problem, conveying the message that ``adaptive data analysis requires more samples than non adaptive data analysis''. In this work we revisit the ADA problem at a foundational level, and ask:
\begin{question}\label{quest:3}
Is there a more fundamental bottleneck for the ADA problem than the number of samples?
\end{question}
Consider a mechanism $\AAA$ that initially gets the {\em full description of the underlying distribution $\PPP$}, but is required to shrink this description into a summary $z$, whose description length is identical to the description length of $t$ samples from $\PPP$. (We identify the {\em space complexity} of $\AAA$ with the size of $z$ in bits.) Afterwards, $\AAA$ needs to answer $k$ adaptive queries using $z$, without additional access to $\PPP$. Does this give $\AAA$ more power over a mechanism that only obtains $t$ samples from $\PPP$?

We show that, in general, the answer is no. Specifically, we show the following theorem.
\begin{theorem}[informal version of Theorem~\ref{thm:ADA}]
Assuming the existence of one-way functions,
then every computationally efficient mechanism that is $(0.1,0.1)$-accurate for $k$ queries must have {\em space complexity} at least $\Omega\left(\sqrt{k}\right)$.
\end{theorem}

In fact, in the formal version of this theorem (see Theorem~\ref{thm:ADA}) we show that the space complexity must be at least $\Omega\left(\sqrt{k}\right)$ times the representation length of domain elements. 
We view this as a significant strengthening of the previous lower bounds for the ADA problem:
it is not that the mechanism did not get enough information about $\PPP$; it is just that it cannot shrink this information in a way that allows for efficiently answering $k$ adaptive queries. This generalizes the negative results of \cite{HardtU14,SteinkeU15}, as sampling $t=\sqrt{k}$ points from $\PPP$ is just one particular way for storing information about $\PPP$.

\subsection{Our techniques}
We obtain our results through a combination of techniques across several research areas including cryptography, privacy, learning theory, communication complexity, and information theory.

\subsubsection{Our techniques: Multi-instance leakage-resilient scheme}
The main cryptographic tool we define and construct is a suitable encryption scheme with multiple keys  that is built to withstand a certain amount of key leakage in a very particular way. Specifically, the scheme consists of $n$ instances of an underlying encryption scheme with independent keys (each of length $\lambda$ bits). The keys are initially given to an adversary who shrinks them down to a summary $z$ containing $s \ll n \cdot \lambda$ bits. After this phase, each instance independently sets an additional parameter, which is public but unknown to the adversary in the initial phase.
Then, the adversary obtains encryptions of plaintexts under the $n$ keys. We require that given $z$ and the public parameters, the plaintexts encrypted with each key remain computationally hidden, except for a small number of the keys (which depends on $s$, but not on $n$). 

We call the scheme a {\em multi-instance leakage-resilient scheme} (or MILR scheme) to emphasize the fact that although the leakage of the adversary is an arbitrary function of all the $n$ keys, the scheme itself is composed of $n$ instances that are completely independent.

The efficiency of the MILR scheme is measured by two parameters: (1) the number of keys under which encryptions are (potentially) insecure, and (2) the loss in the security parameter $\lambda$.
The scheme we construct is optimal in both parameters up to a multiplicative constant factor. First,
encryptions remain hidden for all but $O \left(\frac{s}{\lambda} \right)$ of the keys. This is essentially optimal, as the adversary can define $z$ to store the first $\frac{s}{\lambda}$ keys. Second, we lose a constant factor in the security parameter $\lambda$. An additional advantage of our construction is that its internal parameters do not depend on $s$. If we did allow such a dependency, then in some settings (particularly when $s \leq o(n \cdot \lambda)$) it would be possible to fine-tune the scheme to obtain a multiplicative $1 + o(1)$ loss in the efficiency parameters, but this has little impact on our application. 

Our construction is arguably the most natural one. To encrypt a plaintext with a $\lambda$-bit key after the initial phase, we first apply an extractor (with the public parameter as a seed) to hash it down to a smaller key, which is used to encrypt the plaintext with the underlying encryption scheme. 

The MILR scheme is related to schemes developed in the area of leakage-resilient cryptography (cf.,~\cite{CanettiDHKS00,Dziembowski06,DziembowskiP08,Pietrzak09,StandaertPYQYO10,HazayLWW16,KalaiR19}) where the basic technique of randomness extraction is commonly used.
However, leakage-resilient cryptography mainly deals with resilience of cryptosystems to side-channel attacks, whereas our model is not designed to formalize security against such attacks and has several properties that are uncommon in this domain (such as protecting independent multiple instances of an encryption scheme in a way that inherently makes some of them insecure).
Consequently, the advanced cryptosytems developed in the area of leakage-resilient cryptography are either unsuitable, overly complex, or  inefficient for our purposes.

Despite the simplicity of our construction, our proof that it achieves the claimed security property against leakage is somewhat technical. We stress that we do not rely on hardness assumptions for specific problems, nor assume that the underlying encryption scheme has special properties such as resilience to related-key attacks. Instead, our proof is based on the pre-sampling technique introduced by Unruh~\cite{Unruh07} to prove security of cryptosystems against non-uniform adversaries in the random oracle model. This technique has been recently refined and optimized in~\cite{CorettiDGS18,CorettiDG18,DodisFMT20} based on tools developed in the area of communication complexity~\cite{GoosLM0Z15,KothariMR17}.
The fact that we use the technique to prove security in a computational (rather than information-theoretic) setting seems to require the assumption that the underlying encryption scheme is secure against non-uniform adversaries (albeit this is considered a standard assumption).

\subsubsection{Our techniques: Privacy requires more space}

We design a problem that can be solved non-privately with very small space complexity, but requires a large space complexity with privacy. To achieve this, we lift a known negative result on the {\em sample} complexity of privately solving a specific problem to obtain a {\em space} lower bound for a related problem. 
The problem we start with, for which there exists a sampling lower bound, is the so-called {\em 1-way marginals} problem with parameter $d$. In this problem, our input is a dataset $D\in(\{0,1\}^d)^*$ containing a collection of binary vectors, each of length $d$. Our goal is to output a vector $\vec{a}\in[0,1]^d$ that approximates the {\em average} of the input vectors to within small $L_{\infty}$ error, say to within error $1/10$. That is, we want vector $\vec{a}\in[0,1]^d$ to satisfy: 
$
\left\| \vec{a} - \frac{1}{|D|}\sum_{\vec{x}\in D} \vec{x} \right\|_{\infty} \leq \frac{1}{10}.
$

We say that an algorithm for this problem has {\em sample complexity} $n$ if, for every input dataset of size $n$, it outputs a good solution with probability at least $0.9$. One of the most fundamental results in the literature of differential privacy shows that this problem requires a large dataset:

\begin{theorem}[\cite{BunUV14}, informal]
Every differentially private algorithm for solving the 1-way marginal problem with parameter $d$ must have sample complexity $n=\Omega(\sqrt{d})$.
\end{theorem}

To lift this sampling lower bound into a space lower bound, we consider a problem in which the input dataset contains $n$ {\em keys} $\vec{x}=(x_1,\dots,x_n)$ (sampled from our MILR scheme). The algorithm must then shrink this dataset into a summary $z$ containing $s$ bits. Afterwards, the algorithm gets a ``query'' that is specified by a collection of $n$ ciphertexts, each of which is an encryption of a $d$-bit vector. The desired task is to approximate the average of the {\em plaintext} input vectors. Intuitively, if the algorithm has space $s\ll\sqrt{d}$, then by the properties of our MILR scheme, it can decrypt at most $\approx s\ll\sqrt{d}$ of these $d$-bit vectors, and is hence trying to solve the 1-way marginal problem with fewer than $\sqrt{d}$ samples. We show that this argument can be formalized to obtain a contradiction.

\subsubsection{Our techniques: ADA is about space}

As we mentioned, Hardt and Ullman~\cite{HardtU14} and Steinke and Ullman~\cite{SteinkeU15} showed that the ADA problem requires a large {\em sample} complexity (see Theorem~\ref{thm:HSU1415}). Specifically, they showed that there exists a computationally efficient adversary that causes every efficient mechanism to fail in answering adaptive queries. Recall that the ADA game (see Figure~\ref{fig:ADAintro}) begins with the adversary choosing the underlying distribution.

We lift the negative result of \cite{HardtU14,SteinkeU15} to a {\em space} lower bound. To achieve this, we design an alternative adversary that first samples a large collection of {\em keys} for our MILR scheme, and then defines the target distribution to be uniform on these sampled keys. Recall that in our setting, the mechanism gets an exact description of this target distribution and afterwards it must shrink this description into $s$ bits. However, by the properties of our MILR scheme, this means that the mechanism would only be able to decrypt ciphertexts that correspond to at most $\approx s/\lambda$ of the keys. We then show that the adversary of \cite{HardtU14,SteinkeU15} can be simulated under these conditions, where the ``input sample'' from their setting corresponds to the collection of indices of keys for which the mechanism can decrypt ciphertexts.

\subsection{Additional related work}
\label{app:relatedWork}

Query-to-communication lifting theorems were applied by many works in the domain of communication complexity, e.g.~\cite{goos2015deterministic,GoosLM0Z15,de2016limited,chattopadhyay2019simulation,wu2017raz,goos2020query,garg2018monotone,chattopadhyay2018simulation,loff2019lifting,hatami2018structure,goos2014communication,chattopadhyay2021query,chattopadhyay2019query}. As we mentioned, to obtain our negative result for DP algorithms, we lift the {\em sampling} lower bound of Bun et al.~\cite{BunUV14} for the 1-way marginal problem to a {\em space} lower bound. Many additional sampling lower bounds exist in the DP literature; here we only survey a few of them. 

Dwork et al.~\cite{dwork2009complexity} used {\em traitor-tracing schemes} to prove computational sampling lower bounds for differential
privacy. Their results were extended by Ullman~\cite{ullman2013answering}, who used {\em fingerprinting codes} to construct a novel traitor-tracing scheme and to obtain stronger computational sampling lower bounds for DP. Ullman's construction can be viewed as an encrypted variant of the 1-way marginal problem. Bun et al.~\cite{BunNSV15} and Alon et al.~\cite{AlonLMM19} showed that there are trivial learning tasks that require asymptotically more samples to solve with differential privacy (compared to the non-private sample complexity). 
Feldman~\cite{feldman2020does} and Brown et al.~\cite{brown2021memorization} showed that there are learning problems for which every {\em near optimal} learning algorithm (obtaining near optimal error w.r.t.\ to the number of input {\em samples} it takes) must memorize a large portion of its input data. 
We stress that all of these results are fundamentally different than ours, as they are about sampling rather than space.

Steinhardt et al.~\cite{SteinhardtVW16} asked whether there exists a separation in learning a concept class given additional space constraints and conjectured that the problem of parity learning provided such a separation. Raz
\cite{Raz19} proved the conjecture to be true by showing that although parity learning can be solved in $O(n)$ samples and $O(n^2)$ space, any algorithm for parity learning that uses $o(n^2)$ samples requires an \emph{exponential} number of samples. 
Subsequently, \cite{KolRT17} showed that learning linear-size DNF formulas, linear-size decision trees and logarithmic-size juntas is infeasible under certain super-linear space constraints. 
An active line of work further characterized the limitations of learning concept classes~\cite{MoshkovitzM17,Raz17,MoshkovitzM18,BeameGY18,GargRT18,DaganS18} and distributions subject to space constraints~\cite{ChienLM10,DiakonikolasGKR19,DiakonikolasKPP22}. 
We remark that although the message of higher sample complexity due to space constraints is conceptually similar to our adaptive data analysis result, the problem of learning concept classes is not fundamentally inherent to accurately responding to adaptive queries. 

\subsection{Applications to Communication Complexity}
Finally, our arguments also provide distributional one-way communication complexity lower bounds, which are useful when the goal is to compute a relation with a very low success probability. 
To the best of our knowledge, existing query-to-communication lifting theorems, e.g., \cite{RazM99,goos2015deterministic,GoosLM0Z15}  do not consider the problems and input distributions that we consider here. 
Roughly speaking, we show that if any sampling based protocol for computing a function $f$ requires $k$ samples $(a_1,b_1),\ldots,(a_k,b_k)$, where $a_i\in\{0,1\}^t$ for each $i\in[k]$, then any one-way protocol that computes $f(A,B)$ in this setting must use $\Omega(kt)$ communication.

More precisely, in the two-player one-way communication game, inputs $A$ and $B$ are given to Alice and Bob, respectively, and the goal is for Alice to send a minimal amount of information to Bob, so that Bob can compute $f(A,B)$ for some predetermined function $f$. 
The communication cost of a protocol $\Pi$ is the size of the largest message in bits sent from Alice across all possible inputs $A$ and the (randomized) communication complexity is the minimum communication cost of a protocol that succeeds with probability at least $\frac{2}{3}$. 
In the distributional setting, $A$ and $B$ are further drawn from a known underlying distribution. 

In our distributional setting, suppose Alice has $m$ independent and uniform numbers $a_1,\ldots,a_m$ so that either $a_i\in GF(p)$ for all $i\in[m]$ or $a_i\in GF(2^t)$ for sufficiently large $t$ for all $i\in[m]$ and suppose Bob has $m$ independent and uniform numbers $b_1,\ldots,b_m$ from the same field, either $GF(p)$ or $GF(2^t)$. 
Then for any function $f(\langle a_1,b_1\rangle,\ldots,\langle a_m,b_m\rangle)$, where the dot products are taken over $GF(2)$ or $f(a_1\cdot b_1,\ldots,a_m\cdot b_m)$, where the products are taken over $GF(p)$, has the property that the randomized one-way communication complexity of computing $f$ with probability $\sigma$ is the same as the number of samples from $a_1,\ldots, a_m$ that Alice needs to send Bob to compute $f$ with probability $\sigma-\eps$. 
It is easy to prove sampling lower bounds for many of these problems, sum as $\sum_i a_i\cdot b_i\pmod{p}$ or $\MAJ(\langle a_1, b_1\rangle,\ldots,\langle a_m, b_m\rangle)$, and this immediately translates into communication complexity lower bounds. 
The main intuition for these lower bounds is that the numbers $b_1,\ldots,b_m$ can be viewed as the hash functions that Bob has and thus we can apply a variant of the leftover hash lemma if Bob only has a small subset of these numbers. 
We provide full details in Appendix~\ref{sec:app:cc}. 

\subsection{Paper structure}
The rest of the paper is structured as follows.
The MILR scheme is defined in Section~\ref{sec:MILR}. Our results for differential privacy and adaptive data analysis are described in Sections~\ref{sec:privacy}
and~\ref{sec:ADAbody}, respectively. We construct our MILR scheme in Section~\ref{sec:constructME}, 
and prove its security in Section~\ref{sec:preprocessing}. 
For readability, some of the technical details from these sections are deferred to Appendices~\ref{app:ADAfullProof}, \ref{app:constructME}, and~\ref{app:preprocessing}. 
Appendix~\ref{sec:additionalPrelims} gives standard preliminaries from cryptography. 
Finally, Appendix~\ref{sec:app:cc} details the relationship between sample complexity and communication complexity for one-way protocols. 

\section{Multi-instance Leakage-resilient Scheme} \label{sec:MILR}
We define a {\em multi-instance leakage-resilient scheme} (or MILR scheme) to be a tuple of efficient algorithms $(\Gen,$ $\Param,\Enc,\Dec)$ with the following syntax:
\begin{itemize}
    \item $\Gen$ is a randomized algorithm that takes as input a security parameter $\lambda$ and outputs a $\lambda$-bit secret key. Formally,
    $x\leftarrow\Gen(1^\lambda)$.
    \item $\Param$ is a randomized algorithm that takes as input a security parameter $\lambda$ and outputs a $\poly(\lambda)$-bit public parameter. Formally,
    $p\leftarrow\Param(1^\lambda)$.
    \item $\Enc$ is a randomized algorithm that takes as
    input a secret key $x$, a public parameter $p$, and a
    message $m\in\{0,1\}$, and outputs a
    ciphertext $c\in\{0,1\}^{\poly(\lambda)}$. Formally, $c\leftarrow\Enc(x,p,m)$.
    \item $\Dec$ is a deterministic algorithm that takes as input a secret key $x$, a public parameter $p$, and a ciphertext $c$, and outputs a decrypted message $m'$. If the ciphertext $c$ was an encryption of $m$ under the key $x$ with the parameter $p$, then $m' = m$. Formally, if $c\leftarrow\Enc(x,p,m)$, then $\Dec(x,p,c) = m$ with probability $1$.
\end{itemize}
To define the security of an MILR scheme, let $n \in\N$, let $\vec{x}=(x_1,\dots,x_n)$ be a vector of keys, and let $\vec{p}=(p_1,\dots,p_n)$ be a vector of public parameters (set once for each scheme by invoking $\Param$). Let $J\subseteq[n]$ be a subset, referred to as the ``hidden coordinates''. Now consider a pair of oracles $\oracle_1(\vec{x},\vec{p},J,\cdot,\cdot)$ and $\oracle_0(\vec{x},\vec{p},J,\cdot,\cdot)$ with the following properties.

\begin{enumerate}
	\item $\oracle_1(\vec{x},\vec{p},J,\cdot,\cdot)$ takes as input an index of a key $j\in[n]$ and a message $m$, and returns $\Enc(x_j,p_j,m)$.
	\item $\oracle_0(\vec{x},\vec{p},J,\cdot,\cdot)$ takes the same inputs. If $j\in J$ then the outcome is $\Enc(x_j,p_j,0)$, and otherwise the outcome is $\Enc(x_j,p_j,m)$.
\end{enumerate}

\begin{definition}\label{def:muiltienc}
Let $\lambda$ be a security parameter.
Let $\Gamma:\R\rightarrow\R$ and $\overline{\tau}:\R^2\rightarrow\R$ be real-valued functions. An MILR scheme $(\Gen,\Param,\Enc,\Dec)$ is $(\Gamma,\overline{\tau})$-\emph{secure} against space bounded preprocessing adversaries if the following holds.
\begin{enumerate}
\item[{\bf (1)}] {\bf Multi semantic security:} For every $n=\poly(\Gamma(\lambda))$ and every $\poly(\Gamma(\lambda))$-time adversary $\BBB$ there exists a negligible
function $\negl$ such that
    $$
\left|
\Pr_{\vec{x},\vec{p},\BBB,\Enc}\left[  \BBB^{\oracle_0(\vec{x},\vec{p},[n],\cdot,\cdot)}(\vec{p})=1  \right]
-
\Pr_{\vec{x},\vec{p},\BBB,\Enc}\left[  \BBB^{\oracle_1(\vec{x},\vec{p},[n],\cdot,\cdot)}(\vec{p})=1  \right]
\right| \leq \negl(\Gamma(\lambda)).
    $$
    That is, a computationally bounded adversary that gets the public parameters, but not the keys, cannot tell whether it is interacting with $\oracle_0$ or with $\oracle_1$.
\item[{\bf (2)}] {\bf Multi-security against a bounded preprocessing adversary:} 
For every $n=\poly(\Gamma(\lambda))$, every $s\leq n\cdot\lambda$, and every preprocessing procedure $F:\left(\{0,1\}^{\lambda}\right)^n\rightarrow \{0,1\}^s$ (possibly randomized), 
there exists a random function $J=J(F,\vec{x},z,\vec{p})\subseteq[n]$ that given a collection of keys $\vec{x}$ and public parameters $\vec{p}$, and an element $z\leftarrow F(\vec{x})$, returns a subset of size $|J| \geq \tau:=n-\overline{\tau}(\lambda,s)$ such that for every $\poly(\Gamma(\lambda))$-time algorithm $\BBB$ there exists a negligible
function $\negl$ satisfying
$$
\left|
\Pr_{\substack{\vec{x},\vec{p},\BBB,\Enc\\
z\leftarrow F(\vec{x})\\
J\leftarrow J(F,\vec{x},z,\vec{p}) }}\left[  \BBB^{\oracle_0(\vec{x},\vec{p},J,\cdot,\cdot)}(z,\vec{p})=1  \right]
-
\Pr_{\substack{\vec{x},\vec{p},\BBB,\Enc\\
z\leftarrow F(\vec{x})\\
J\leftarrow J(F,\vec{x},z,\vec{p}) }}\left[  \BBB^{\oracle_1(\vec{x},\vec{p},J,\cdot,\cdot)}(z,\vec{p})=1  \right]
\right| \leq \negl(\Gamma(\lambda)).
$$
That is, even if $s$ bits of our $n$ keys were leaked (computed by the preprocessing function $F$ operating on the keys), then still encryptions w.r.t.\ the keys of $J$ are computationally indistinguishable.
\end{enumerate}

\end{definition}

\begin{remark}

When $\Gamma$ is the identity function, we simply say that the scheme is $\overline{\tau}$-secure. Note that in this case, security holds against all adversaries with runtime polynomial in the security
parameter $\lambda$. We will further assume the existence of a sub-exponentially secure encryption scheme. By that we mean that there exists a constant $\nu>0$ such that the scheme is $(\Gamma,\overline{\tau})$-secure for $\Gamma(\lambda)=2^{{\lambda}^{\nu}}$. That is, we assume the existence of a scheme in which security holds against all adversaries with runtime polynomial in $2^{{\lambda}^{\nu}}$.
\end{remark}

In Sections~\ref{sec:constructME} and~\ref{sec:preprocessing} we show the following theorem.

\begin{theorem}
\label{thm:multi_first_statement}
Let $\Omega(\lambda) \leq \Gamma(\lambda) \leq 2^{o(\lambda)}$.
If there exists a $\Gamma(\lambda)$-secure encryption scheme against non-uniform adversaries then there exists an MILR scheme that is $(\Gamma(\lambda),\overline{\tau})$-\emph{secure} against space bounded non-uniform preprocessing adversaries, where $\overline{\tau}(\lambda,s) = \frac{2s}{\lambda} + 4$.
\end{theorem}

\section{Space Hardness for Differential Privacy} \label{sec:privacy}

Consider an algorithm that is instantiated on a dataset $D$ and then aims to answer a query w.r.t.\ $D$. We say that such an algorithm has space $s$ if, before it gets the query, it shrinks $D$ to a {\em summary} $z$ containing at most $s$ bits. Then, when obtaining the query, the algorithm answers it using only the summary $z$ (without additional access to the original dataset $D$).

Let $(\Gen,\Param,\Enc,\Dec)$ be an MILR scheme with security parameter $\lambda$. In the $(\lambda,d)$-decoded average (DA) problem with sample complexity $n$, the input dataset contains $n$ keys, that is $D=(x_1,\dots,x_n)\in\left(\{0,1\}^\lambda\right)^n$. A query $q=\left((p_1,c_1),\dots,(p_n,c_n)\right)$ is specified using $n$ pairs $(p_i,c_i)$ where $p_i$ is a public parameter and $c_i$ is a ciphertext, which is an encryption of a binary vector of length $d$. The goal is to release a vector $\vec{a}=(a_1,\dots,a_d)\in[0,1]^d$ that approximates the ``decrypted average vector (dav)'', defined as 
$
{\rm dav}_q (D)=
\frac{1}{n}\sum_{i=1}^n \Dec(x_i,p_i,c_i).
$

\begin{definition}
Let $\AAA=(\AAA_1,\AAA_2)$ be an algorithm where 
$\AAA_1:\left(\{0,1\}^\lambda\right)^n\rightarrow\{0,1\}^s$ is the preprocessing procedure that summarizes a dataset of $n$ keys into $s$ bits, and where $\AAA_2$ is the ``response algorithm'' that gets the outcome of $\AAA_1(D)$ and a query $q$. We say that $\AAA$ solves the DA problem if with probability at least $9/10$ the output is a vector $\vec{a}$ satisfying
$
\left\|\vec{a} - {\rm dav}_q (D) \right\|_{\infty} \leq\frac{1}{10}.
$
\end{definition}

Without privacy considerations, the DA problem is almost trivial. Specifically,
\begin{lemma}\label{lem:DAisEasyNonPrivately}
The $(\lambda,d)$-DA problem can be solved efficiently using space $s=O\left( \lambda\log(d) \right)$.
\end{lemma}

\begin{proof}
The preprocessing algorithm $\AAA_1$ samples $O\left( \log d \right)$ of the input keys. Algorithm $\AAA_2$ then gets the query $q$ and estimates the ${\rm dav}$ vector using the sampled keys. The lemma then follows by the Chernoff bound.
\end{proof}

\begin{remark}\label{rem:streamingModel}
As we mentioned, to simplify the presentation, in our computational model we identify the space complexity of algorithm $\AAA=(\AAA_1,\AAA_2)$ with the size of the output of algorithm $\AAA_1$. We remark, however, that our separation extends to a streaming model where both $\AAA_1$ and $\AAA_2$ are required to have small space. To see this, note that algorithm $\AAA_1$ in the above proof already has small space (and not just small output length), as it merely samples $O(\log d)$ keys from its input dataset. We now analyze the space complexity of $\AAA_2$, when it reads the query $q$ in a streaming fashion. Recall that the query $q$ contains $n$ public parameters $p_1,\dots,p_n$ and $n$ ciphertexts $c_1,\dots,c_n$, where each $c_i$ is 
an encryption of a $d$-bit vector, call it $y_i\in\{0,1\}^d$. To allow $\AAA_2$ to read $q$ using small space, we order it as follows:
$
q=(p_1,\dots,p_n),(c_{1,1},\dots,c_{n,1}),\dots,(c_{1,d},\dots,c_{n,d})\triangleq q_0\circ q_1\circ\dots\circ q_d.
$
Here $c_{i,j}=\Enc(x_i,p_i,y_i[j])$ is an encryption of the $j$th bit of  $y_i$ using key $x_i$ and public parameter $p_i$. Note that the first part of the stream, $q_0$, contains the public parameters, and then every part $q_j$ contains encryptions of the $j$th bit of each of the $n$ input vectors. With this ordering of the query, algorithm $\AAA_2$ begins by reading $q_0$ and storing the $O(\log d)$ public parameters corresponding to the keys that were stored by $\AAA_1$. Then, for every $j\in[d]$, when reading $q_j$, algorithm $\AAA_2$ estimates the average of the $j$th coordinate using the sampled keys. Algorithm $\AAA_2$ then outputs this estimated value, and proceeds to the next coordinate. This can be implemented using space complexity $\poly(\lambda\log(d))$.
\end{remark}

So, without privacy constraints, the DA problem can be solved using small space. We now show that, assuming that the input dataset is large enough, the DA problem can easily be solved with differential privacy using {\em large} space. Specifically,

\begin{lemma}\label{lem:DAisSolvablePrivately}
There is a computationally efficient $(\eps,\delta)$-differentially private algorithm that solves the $(\lambda,d)$-DA problem using space $s=O\left( \frac{1}{\eps}\cdot \sqrt{d\cdot\log(\frac{1}{\delta})}\cdot \lambda\cdot \log d \right)$, provided that the size of the input dataset satisfies $n=\Omega(s)$ (large enough).
\end{lemma}

\begin{proof}
The preprocessing algorithm $\AAA_1$ samples $\approx \sqrt{d}$ of the keys. By standard composition theorems for differential privacy~\cite{DworkRV10}, this suffices for the response algorithm $\AAA_2$ to privately approximate each of the $d$ coordinates of the target vector.
\end{proof}

Thus the DA problem can be solved non-privately using small space, and it can be solved privately using large space. Our next goal is to show that large  space is indeed necessary to solve this problem privately. Before showing that, we introduce several additional preliminaries on computational differential privacy and on fingerprinting codes.

\subsection{Preliminaries on computational differential privacy and fingerprinting codes}

Computational differential privacy was defined by Beimel et al~\cite{BeimelNO08} and Mironov et al.~\cite{MironovPRV09}. Let $\AAA$ be a randomized algorithm (mechanism) that operates on datasets.
Computational differential privacy is defined via a two player game between a {\em challenger} and an adversary, running a pair of algorithms $(Q,T)$.
The game begins with the adversary $Q$ choosing a pair of neighboring datasets $(D_0,D_1)$ of size $n$ each, as well as an arbitrary string $r$ (which we think of as representing its internal state). Then the challenger samples a bit $b$ and applies $\AAA(D_b)$ to obtain an outcome $a$. Then $T(r,\cdot)$ is applied on $a$ and tries to guess $b$. %
Formally,

\begin{definition}\label{def:CompDP}
Let $\lambda$ be a security parameter, let $\eps$ be a constant, and let $\delta:\R\rightarrow\R$ be a function. 
A randomized algorithm $\AAA:X^*\rightarrow Y$ is {\em  $(\eps,\delta)$-computationally differentially private (CDP)} if for every $n=\poly(\lambda)$ and every non-uniform $\poly(\lambda)$-time adversary $(Q,T)$ there exists a negligible function $\negl$ such that
$$
\Pr_{\substack{(r,D_0,D_1)\leftarrow Q\\ \AAA, T}}[T(r,\AAA(D_0))=1]\leq e^{\eps}\cdot\Pr_{\substack{(r,D_0,D_1)\leftarrow Q\\ \AAA, T}}[T(r,\AAA(D_1))=1]+\delta(n)+\negl(\lambda).
$$
\end{definition}

\begin{definition}
Let $\eps$ be a constant and let $\delta=\delta(\lambda)$ be a function.
Given two probability ensembles $\XXX=\{X_{\lambda}\}_{\lambda\in\N}$ and $\YYY=\{Y_{\lambda}\}_{\lambda\in\N}$ we write  $\XXX\approx_{\eps,\delta}\YYY$ if for
every non-uniform probabilistic polynomial-time distinguisher $D$ there exists a negligible
function $\negl$ such that
$
\Pr_{x\leftarrow X_{\lambda}}\left[D(x)=1\right]
\leq e^{\eps}\cdot\Pr_{y\leftarrow Y_{\lambda}}\left[D(y)=1\right] + \delta(\lambda) + \negl(\lambda),$ and vice versa.
\end{definition}

We recall the concept of {\em fingerprinting codes}, which was introduced by Boneh and Shaw~\cite{boneh1998collusion} as a cryptographic tool for watermarking digital content. 
Starting from the work of Bun, Ullman, and Vadhan \cite{BunUV14}, fingerprinting codes have played a key role in proving lower bounds for differential privacy in various settings.

A (collusion-resistant) fingerprinting code is a scheme for distributing codewords $w_1,\cdots, w_n$ to $n$ users that can be uniquely traced back to each user. 
Moreover, if a group of (at most $k$) users combines its codewords into a pirate codeword $\hat{w}$, then the pirate codeword can still be traced back to one of the users who contributed to it. 
Of course, without any assumption on how the pirates can produce their combined codeword, no secure tracing is possible. 
To this end, the pirates are constrained according to a {\em marking assumption}, which asserts that the combined codeword must agree with at least one of the  ``real'' codewords in each position. 
Namely, at an index $j$ where $w_{i}[j] = b$ for every $i\in[n]$, the pirates are constrained to output $\hat{w}$ with $\hat{w}[j] = b$ as well.\footnote{We follow the formulation of the marking assumption as given by \cite{SteinkeU15}, which is a bit different than the commonly considered one.}

\begin{definition}[\cite{boneh1998collusion,tardos2008optimal}]
A $k$-collusion resilient fingerprinting code of length $d$ for $n$ users with failure probability $\gamma$, or $(n,d,k,\gamma)$-FPC in short, is a pair of random variables $C\in\{0,1\}^{n\times d}$ and 
${\rm Trace}:\{0,1\}^d\rightarrow2^{[n]}$ such that the following holds. For all adversaries $P:\{0,1\}^{k\times d}\rightarrow\{0,1\}^d$ and $S\subseteq[n]$ with $|S|=k$, 
$$
\Pr_{C,{\rm Trace},P}\left[ \left( \forall 1\leq j\leq d \;\; \exists i\in[n] \text{ s.t.\ } P(C_S)[j]=c_i[j]  \right) \wedge \left( {\rm Trace}(P(C_S))=\emptyset \right)  \right]\leq\gamma,
$$
and 
$$
\Pr_{C,{\rm Trace},P}\left[
{\rm Trace}(P(C_S))\cap([n]\setminus S) \neq\emptyset
\right]\leq\gamma,
$$
 where $C_S$ contains the rows of $C$ given by $S$.
\end{definition}

\begin{remark}\label{rem:marking}
As mentioned, the condition $\left\{\forall 1\leq j\leq d \;\; \exists i\in[n] \text{ s.t.\ } P(C_S)[j]=c_i[j]\right\}$ is called the ``marking assumption''. The second condition is called the ``small probability of false accusation''. Hence, if the adversary $P$ guarantees that its output satisfies the marking assumption, then with probability at least $1-2\gamma$ it holds that algorithm {\rm Trace} outputs an index $i\in S$.
\end{remark}

\begin{theorem}[\cite{boneh1998collusion,tardos2008optimal,SteinkeU15}]\label{thm:FPC}
For every $1\leq k\leq n$ there is a $k$-collusion-resilient fingerprinting code of length $d=O(k^2 \cdot \log n )$ for $n$ users with failure probability $\gamma= \frac{1}{n^2}$ and an efficiently computable Trace function.
\end{theorem}

We remark that there exist both adaptive and non-adaptive constructions of fingerprinting codes with the guarantees of Theorem~\ref{thm:FPC}; we use the non-adaptive variant. 

\subsection{A negative result for the DA problem}

Our main negative result for space bounded differentially private algorithms is the following.

\begin{theorem}\label{thm:DPhard}
Let $\Pi=(\Gen,\Param,\Enc,\Dec)$ be an MILR scheme with security parameter $\lambda$ that is  $(\Gamma,\overline{\tau})$-secure against space bounded preprocessing adversaries. Let $d\leq\poly(\Gamma(\lambda))$ and $n\leq\poly(\Gamma(\lambda))$ be functions of $\lambda$. 
Let $\eps$ be a constant and let $\delta\leq\frac{1}{4n(e^{\eps}+1)}=\Theta(\frac{1}{n})$. 
For every $\poly(\Gamma(\lambda))$-time $(\eps,\delta)$-CDP algorithm for the $(\lambda,d)$-DA problem with sample complexity $n$ and space complexity $s$ it holds that
$
\overline{\tau}(\lambda,s)=\Omega\left(\sqrt{\frac{d}{\log n}}\right).
$
\end{theorem}

\begin{proof}
Let $\AAA=(\AAA_1,\AAA_2)$ be a $\poly(\Gamma(\lambda))$-time CDP algorithm for the $(\lambda,d)$-DA problem using sample complexity $n=\poly(\Gamma(\lambda))$ and space complexity $s$. 
Denote $\overline{\tau}=\overline{\tau}(\lambda,s)$, and assume towards a  contradiction that $\overline{\tau}=O\left(\sqrt{\frac{d}{\log n}}\right)$ (small enough).
We construct the following adversary $\BBB$ to an $\left(n{+}1,d,\overline{\tau},\frac{1}{n^2}\right)$-FPC (such a code is guaranteed to exist by Theorem~\ref{thm:FPC} and by the contradictory assumption). 

\begin{enumerate}
    \item The input is $n$ codewords $w_1,\dots,w_n\in\{0,1\}^d$.
    \item Sample $n$ keys $x_1,\dots,x_n\sim\Gen(1^\lambda)$.
    \item Let $z\leftarrow\AAA_1(x_1,\dots,x_n)$.
    \item Sample $n$ public parameters $p_1,\dots,p_n\sim\Param(1^\lambda)$.
    \item For $i\in[n]$ let $c_i\leftarrow\Enc(x_i,p_i,w_i)$.
    \item Let $\vec{a}\leftarrow\AAA_2(z, (p_1,c_1),\dots, (p_n,c_n) )$.
    \item Output $\vec{a}$, after rounding its coordinates to $\{0,1\}$.
\end{enumerate}

We think of $\BBB$ as an adversary to an FPC, and indeed, its input is a collection of codewords and its output is a binary vector of length $d$. Observe that if $\AAA$ solves the DA problem (i.e., approximates the decrypted average vector), then for every coordinate, the outcome of $\BBB$ must agree with at least one of the input codewords, namely, it satisfies the marking assumption (see Remark~\ref{rem:marking}).

\begin{remark}
Before we proceed with the formal proof, we give an overview of its structure (this remark can be ignored, and is not needed for the formal proof). Informally, we will show that 
\begin{enumerate}
    \item[(1)] Algorithm $\BBB$ is computationally differentially private w.r.t.\ the collection of {\em codewords} (even though our assumption on $\AAA$ is that it is private w.r.t.\ the {\em keys}).
    \item[(2)] Leveraging the properties of the MILR scheme, we will show that $\BBB$ must effectively ignore most of its inputs, except for at most $\overline{\tau}$ codewords. This means that $\BBB$ is effectively an FPC adversary that operates on only $\overline{\tau}$ codewords (rather than the $n$ codewords it obtains as input).
    \item[(3)] A known result in the literature of differential privacy (starting from \cite{BunUV14}) is that a successful FPC adversary cannot be differentially private, because this would contradict the fact that the tracing algorithm is able to recover one of its input points. Our gain here comes from the fact that $\BBB$ only uses (effectively) $\overline{\tau}$ codewords, and hence, in order to get a contradiction, it suffices to use an FPC with a much shorter codeword-length (only $\approx \overline{\tau}^2$ instead of $\approx n^2$). This will mean that the hardness of the DA problem depends on the space of $\AAA$ (which controls $\overline{\tau}$) rather than the size of the input (which is $n$).
\end{enumerate}
\end{remark}

Recall that $\AAA=(\AAA_1,\AAA_2)$ is computationally differentially private w.r.t.\ the {\em keys}. 
We first show that $\BBB$ is computationally differentially private w.r.t.\ the {\em codewords}. %
To this end, let $Q$ be an adversary (as in Definition~\ref{def:CompDP}) that outputs a pair of neighboring datasets $(\vec{w},\vec{w'})$, each containing $n$ codewords, together with a state $r$. Given $(\vec{w},\vec{w'})$, we write $\ell=\ell(\vec{w},\vec{w'})\subseteq[n]$ to denote the index on which $\vec{w},\vec{w'}$ differ. We also write $x_0$ to denote another key, independent of the keys $x_1,\dots,x_n$ sampled by algorithm $\BBB$. 
By the privacy guarantees of algorithm $\AAA$ and by the semantic security of the encryption scheme (see Definition~\ref{def:muiltienc}) we have that
\begin{align*}
\langle r,\BBB&(\vec{w}) \rangle\equiv\\
    \equiv&\; 
    \langle r,\AAA_2\left( \AAA_1(x_1,...,x_{\ell},...,x_n) , \vec{p} ,  \Enc(x_1,p_1,w_1), ..., \Enc(x_{\ell},p_{\ell},w_{\ell}), ..., \Enc(x_n,p_n,w_n) \right)\rangle\\
    \approx_{(\eps,\delta)}&\; 
    \langle r,\AAA_2\left( \AAA_1(x_1,...,\red{x_{0}},...,x_n) , \vec{p} ,  \Enc(x_1,p_1,w_1),..., \Enc(x_{\ell},p_{\ell},w_{\ell}), ..., \Enc(x_n,p_n,w_n) \right)\rangle\\
    \equiv_c&\;  
    \langle r,\AAA_2\left( \AAA_1(x_1,...,x_{0},...,x_n) , \vec{p} ,  \Enc(x_1,p_1,w_1),..., \Enc(x_{\ell},p_{\ell},\red{w'_{\ell}}), ..., \Enc(x_n,p_n,w_n) \right)\rangle\\
    \approx_{(\eps,\delta)}&\; 
    \langle r,\AAA_2\left( \AAA_1(x_1,...,\red{x_{\ell}},...,x_n) , \vec{p} ,  \Enc(x_1,p_1,w_1),..., \Enc(x_{\ell},p_{\ell},w'_{\ell}),..., \Enc(x_n,p_n,w_n) \right)\rangle\\
    \equiv&\; 
    \langle r,\BBB(\vec{w'})\rangle.
\end{align*}
So algorithm $\BBB$ is $(2\eps,(e^{\eps}+1)\delta)$-computationally differentially private. Now consider the following variant of algorithm $\BBB$, denoted as $\hat{\BBB}$. The modifications from $\BBB$ are marked in red.

\begin{enumerate}
    \item The input is $n$ codewords $w_1,\dots,w_n\in\{0,1\}^d$.
    \item Sample $n$ keys $x_1,\dots,x_n\sim\Gen(1^\lambda)$.
    \item Let $z\leftarrow\AAA_1(x_1,\dots,x_n)$.
    \item Sample $n$ public parameters $p_1,\dots,p_n\sim\Param(1^\lambda)$.
    
     \item \red{Let $J\leftarrow J(\AAA_1,\vec{x},z,\vec{p})\subseteq[n]$ be the subset of coordinates guaranteed to exist by Definition~\ref{def:muiltienc}, of size $|J|=n-\overline{\tau}$.}
     \item \red{For $i\in J$ let $c_i\leftarrow\Enc(x_i,p_i,0)$.}
    
    \item For $i\in[n]\setminus J$ let $c_i\leftarrow\Enc(x_i,p_i,w_i)$.
    \item Let $\vec{a}\leftarrow\AAA_2(z, (p_1,c_1),\dots, (p_n,c_n) )$.
    \item Output $\vec{a}$, after rounding its coordinates to $\{0,1\}$.
\end{enumerate}

\begin{remark}
Observe that algorithm $\hat{\BBB}$ is not necessarily computationally efficient, since computing $J(\AAA_1,\vec{x},z,\vec{p})$ might not be efficient. Nevertheless, as we next show, this still suffices to obtain a contradiction and complete the proof of the lower bound. Specifically, we will show that $\hat{\BBB}$ is computationally differentially private (w.r.t.\ the codewords) and that it is a successful adversary to the FPC. This will lead to a contradiction, even if $\hat{\BBB}$ itself is a non-efficient mechanism.
\end{remark}

We now show that, by the multi-security of the MILR scheme (see Definition~\ref{def:muiltienc}), the outcome distributions of $\BBB$ and $\hat{\BBB}$ are computationally indistinguishable. Specifically, we want to show that for every efficient adversary $(Q,T)$, as in Definition~\ref{def:CompDP}, it holds that
$$
|\Pr[T(r,\BBB(\vec{w}))=1]-\Pr[T(r,\hat{\BBB}(\vec{w}))=1]|\leq\negl(\Gamma(\lambda)).
$$
Note that here both expressions are with the same dataset $\vec{w}$ (without the neighboring dataset $\vec{w'}$).
To show this, consider the following algorithm, denoted as $\WWW$, which we view as an adversary to the MILR scheme. This algorithm has only oracle access to encryptions, via an oracle $\oracle$.

\begin{enumerate}
    \item The input of $\WWW$ is $n$ codewords $w_1,\dots,w_n\in\{0,1\}^d$, an element $z$ (supposedly computed by $\AAA_1$), and a collection of $n$ public parameters $p_1,\dots,p_n$.
    \item For $i\in[n]$ let $c_i\leftarrow\oracle(i,w_i)$.
    \item Let $\vec{a}\leftarrow\AAA_2(z, (p_1,c_1),\dots, (p_n,c_n) )$.
    \item Output $\vec{a}$, after rounding its coordinates to $\{0,1\}$.
\end{enumerate}

Now by the multi-security of the MILR scheme we have
\begin{align*}
&\left|\Pr_{Q,T,\BBB}[T(r,\BBB(\vec{w}))=1]-\Pr_{Q,T,\BBB}[T(r,\hat{\BBB}(\vec{w}))=1]\right|\\[1em]
&=\left|\Pr_{\substack{\vec{x},\vec{p},\WWW,Q,T,\Enc\\
z\leftarrow \AAA_1(\vec{x})\\
J\leftarrow J(\AAA_1,\vec{x},z,\vec{p}) }}
[Q(r,\WWW^{\oracle_1(\vec{x},\vec{p},J,\cdot,\cdot)}(\vec{w},z,\vec{p}))=1]-\Pr_{\substack{\vec{x},\vec{p},\WWW,Q,T,\Enc\\
z\leftarrow \AAA_1(\vec{x})\\
J\leftarrow J(\AAA_1,\vec{x},z,\vec{p}) }}
[Q(r,\WWW^{\oracle_0(\vec{x},\vec{p},J,\cdot,\cdot)}(\vec{w},z,\vec{p}))=1]\right|\\[1em]
&\leq\negl(\Gamma(\lambda)).
\end{align*}

So we have that $\hat{\BBB}\equiv_c\BBB$ and we have that $\BBB$ is $(2\eps,(e^{\eps}+1)\delta)$-computationally differentially private. Hence, algorithm $\hat{\BBB}$ is also $(2\eps,(e^{\eps}+1)\delta)$-computationally differentially private (w.r.t.\ the input codewords). Observe that algorithm $\hat{\BBB}$ ignores all but $N-|J|=\overline{\tau}$ of the codewords, and furthermore, the choice of which codewords to ignore is independent of the codewords themselves. 
Now consider the following thought experiment.
\begin{enumerate}
    \item Sample a codebook $w_0,w_1,\dots,w_n$ for the fingerprinting code.
    \item Run $\hat{\BBB}$ on $(w_1,\dots,w_n)$.
    \item Run ${\rm Trace}$ on the outcome of $\hat{\BBB}$ and return its output.
\end{enumerate}
As $\hat{\BBB}$ ignores all but $\overline{\tau}$ codewords, by the properties of the FPC, with probability at least $1-\frac{1}{n^2}\geq\frac{1}{2}$ the outcome of ${\rm Trace}$ is a coordinate of a codeword that $\hat{\BBB}$ did not ignore, and in particular, it is a coordinate between $1$ and $n$. Therefore, there must exist a coordinate $i^*\neq 0$ that is output by this thought experiment with probability at least $\frac{1}{2n}$. Now consider the following modified thought experiment.
\begin{enumerate}
    \item Sample a codebook $w_0,w_1,\dots,w_n$ for the fingerprinting code.
    \item Run $\hat{\BBB}$ on $(w_1,\dots,w_{i^*-1},\red{w_0},w_{i^*+1},\dots,w_n)$.
    \item Run ${\rm Trace}$ on the outcome of $\hat{\BBB}$ and return its output.
\end{enumerate}
As $\hat{\BBB}$ is computationally differentially private and as ${\rm Trace}$ is an efficient algorithm, the probability of outputting $i^*$ in this second thought experiment is roughly the same as in the previous thought experiment, specifically, at least 
\begin{align*}
e^{-2\eps}\left(\frac{1}{2 n}-(e^{\eps}+1)\delta-\negl(\Gamma(\lambda))\right)\geq
e^{-2\eps}\left(\frac{1}{4 n}-\negl(\Gamma(\lambda))\right)=\Omega\left(\frac{1}{n}\right).
\end{align*}
However, by the guarantees of the FPC (small probability of false accusation), in the second experiment the probability of outputting $i^*$ should be at most $\frac{1}{n^2}$. This is a contradiction to the existence of algorithm $\AAA$.
\end{proof}

The following corollary follows by instantiating Theorem~\ref{thm:DPhard} with our MILR scheme, as specified in Theorem~\ref{thm:multi_first_statement}.

\begin{corollary}\label{cor:finalDP}
Let $\Omega(\lambda) \leq \Gamma(\lambda) \leq 2^{o(\lambda)}$, and let $d\leq\poly(\Gamma(\lambda))$.
If there exists a $\Gamma(\lambda)$-secure encryption scheme against non-uniform adversaries then there exists an MILR scheme such that the corresponding $(\lambda,d)$-DA problem requires large space to be solved privately. Specifically, 
let $n\leq\poly(\Gamma(\lambda))$, 
let $\eps$ be a constant, and let $\delta\leq\frac{1}{4n(e^{\eps}+1)}=\Theta(\frac{1}{n})$. 
Every 
$\poly(\Gamma(\lambda))$-time $(\eps,\delta)$-CDP algorithm for the $(\lambda,d)$-DA problem with sample complexity $n$ must have space complexity
$
s=\Omega\left(\lambda\cdot\sqrt{\frac{d}{\log n}}\right).
$
\end{corollary}

\section{Space Hardness for Adaptive Data Analysis (ADA)}\label{sec:ADAbody}

Consider a mechanism that first gets as input a sample containing $t$ i.i.d.\ samples from some underlying (unknown) distribution $\DDD$, and then answers $k$ {\em adaptively chosen} statistical queries w.r.t.\ $\DDD$. Importantly, the answers must be accurate w.r.t.\ the underlying distribution and not just w.r.t.\ the empirical sample. The challenge here is that as the queries are being chosen adaptively, the interaction might quickly lead to {\em overfitting}, i.e., result in answers that are only accurate w.r.t.\ the empirical sample and not w.r.t.\ the underlying distribution. 
This  fundamental problem, which we refer to as the ADA problem, was introduced by Dwork et al.~\cite{dwork2015preserving} who connected it to differential privacy and showed that differential privacy can be used as a countermeasure against overfitting. Intuitively, overfitting happens when  answers reveal properties that are specific to the input sample, rather than to the general population, and this is exactly what differential privacy aims to protect against.

Hardt, Steinke, and Ullman~\cite{HardtU14,SteinkeU15} showed negative results for the ADA problem. Specifically, they showed that given $t$ samples, it is computationally hard to answer more than $k=O(t^2)$ adaptive queries. We show that the hardness of the ADA problem is actually more fundamental; it is, in fact, a result of a space bottleneck rather than a sampling bottleneck. Informally, we show that the same hardness result continues to hold even if in the preprocessing stage the mechanism is given the full description of the underlying distribution $\DDD$, and is then required to store only a limited amount of information about it (an amount that equals the representation length of $t$ samples from $\DDD$). So it is not that the mechanism did not get enough information about $\DDD$; it is just that it cannot shrink this information in a way that supports $t^2$ adaptive queries. This generalizes the negative results of \cite{HardtU14,SteinkeU15}, as sampling $t$ points from $\DDD$ is just one particular way of trying to store information about $\DDD$.

Consider \texttt{AdaptiveGameSpace}, where the mechanism initially gets the full description of the underlying distribution, but it must shrink it into an $s$-bit summary $z$. To emphasize that the mechanism does not have additional access to the underlying distribution, we think about it as two mechanisms $\AAA=(\AAA_1,\AAA_2)$ where $\AAA_1$ computes the summary $z$ and where $\AAA_2$  answers queries given $z$. We consider $s=|z|$ as the space complexity of such a mechanism $\AAA$.

\begin{algorithm}[t]\caption{$\texttt{AdaptiveGameSpace}(\AAA{=}(\AAA_1,\AAA_2),\BBB,s,k)$}

\begin{enumerate}
    \item The adversary $\BBB$ chooses a distribution $\DDD$ over a domain $\XXX$.
    \item The mechanism $\AAA_1$ gets $\DDD$ and summarizes it into $s$ bits, denoted as $z$. 
    \item The mechanism $\AAA_2$ is instantiated with $z$.
    \item For round $i=1,2,\dots,k$:
    \begin{enumerate}
        \item The adversary $\BBB$ specifies a query $q_i:\XXX\rightarrow\left\{-1,0,1\right\}$
        \item The mechanism $\AAA_2$ obtains $q_i$ and responds with an answer $a_i\in[-1,1]$ 
        \item $a_i$ is given to $\AAA$
    \end{enumerate}
    \item The outcome of the  game is one if $\exists i$ s.t.\ $\left|a_i - \E_{y\sim\DDD}[q_i(y)] \right|>1/10$, and zero otherwise.
\end{enumerate}
\end{algorithm}

Our main theorem in the context of the ADA problem is the following.

\begin{theorem}\label{thm:ADA}
Let $\Omega(\lambda) \leq \Gamma(\lambda) \leq 2^{o(\lambda)}$, and let $k\leq\poly(\Gamma(\lambda))$.
If there exists a $\Gamma(\lambda)$-secure encryption scheme against non-uniform adversaries then 
there exists a $\poly(\Gamma(\lambda))$-time adversary $\BBB$ such that the following holds. Let $\AAA{=}(\AAA_1,\AAA_2)$ be a $\poly(\Gamma(\lambda))$-time mechanism with space complexity $s\leq O\left(\lambda\cdot\sqrt{k}\right)$ (small enough). Then, 
$
\Pr[\texttt{AdaptiveGameSpace}(\AAA,\BBB,s,k)=1]>\frac{2}{3}.
$

Furthermore, the underlying distribution defined by the adversary $\BBB$ can be fully described using $O(\sqrt{k} \cdot \lambda)$ bits, 
it is sampleable in $\poly(\Gamma(\lambda))$-time, and elements sampled from this distribution can be represented using $O(\lambda+\log(k))$ bits.
\end{theorem}

In a sense, the ``furtheremore'' part of the theorem shows that the distribution chosen by our adversary is not too complex. Specifically, our negative result continues to hold even if the space of the mechanism is {\em linear} in the full description length of the underlying distribution (in a way that allows for efficiently sampling it). If the space of the mechanism was just a constant times bigger, it could store the full description of the underlying distribution and answer an unbounded number of adaptive queries. 
The formal proof of Theorem~\ref{thm:ADA} is deferred to Appendix~\ref{app:ADAfullProof}. Here we only provide an informal (and overly simplified) proof sketch.

\subsection{Informal proof sketch}
Let $k$ denote the number of queries that the adversary makes. 
Our task is to show that there is an adversary that fails every efficient mechanism $\AAAspace$ that plays in \texttt{AdaptiveGameSpace}, provided that it uses space $s\ll\sqrt{k}$. What we know from \cite{HardtU14,SteinkeU15} is that there is an adversary $\BBBsamp$ that fails every efficient mechanism that plays in the standard ADA game (the game specified in Figure~\ref{fig:ADAintro}), provided that its sample complexity is $t\ll\sqrt{k}$. 
We design an adversary $\BBBspace$ that plays in \texttt{AdaptiveGameSpace} in a way that emulates $\BBBsamp$. We now elaborate on the key points in the construction of $\BBBspace$, and their connection to $\BBBsamp$.

Recall that both games begin with the adversary specifying the underlying distribution. A useful fact about the adversary $\BBBsamp$ (from \cite{HardtU14,SteinkeU15}) is that the distribution it specifies is uniform on a small set of points of size $n=\Theta(t)$ 
(these $n$ points are unknown to the mechanism that $\BBBsamp$ plays against).  
Our adversary, $\BBBspace$, first samples $n$ independent keys $(x_1,\dots,x_n)$ from our MILR scheme, %
and then defines the target distribution $\DDDspace$ to be uniform over the set $\{(j,x_j)\}_{j\in[n]}$. Recall that in \texttt{AdaptiveGameSpace} this target distribution is given to the mechanism $\AAAspace$, who must shrink it into a summary $z$ containing $s$ bits. After this stage, by the security of our MILR scheme, there should exist a large set  $J\subseteq[n]$ of size $|J|=n-\overline{\tau}$ corresponding to keys uncompromised by $\AAAspace$. Denote $I=[n]\setminus J$.

Our adversary $\BBBspace$ now emulates $\BBBsamp$ as follows. First, let $\DDDsamp$ denote the target distribution chosen by $\BBBsamp$, and let $m_1,\dots,m_n$ denote its support. 
Our adversary then samples $n$ public parameters $p_1,\dots,p_n$, and encrypts every point in the support $m_j$ using its corresponding key and public parameter. Specifically,  $c_j\leftarrow\Enc(x_j,p_j,m_j)$. {\small \gray{\% This is an over-simplification. For technical reasons, the actual construction is somewhat different.}}

Now, for every query $q$ specified by $\BBBsamp$, our adversary outputs the query $f_q$ defined by $f_q(j,x) = q(\Dec(x, p_j , c_j))$. Our adversary then obtains an answer $a$ from the mechanism $\AAAspace$, and feeds $a$ to $\BBBsamp$.
Observe that 
the ``true'' value of $f_q$ w.r.t.\ $\DDDspace$ is the same as the ``true'' value of $q$ w.r.t.\ $\DDDsamp$. Therefore, if $\AAAspace$ maintains accuracy in this game against our adversary $\BBBspace$, then in the emulation that $\BBBspace$ runs internally we have that $\BBBspace$ maintains utility against $\BBBsamp$. 
Intuitively, we would like to say that this leads to a contradiction, since $\BBBsamp$ fails every efficient mechanism it plays against. But this is not accurate, because $\BBBspace$ saw the full description of the target distribution $\DDDsamp$, and $\BBBsamp$ only fools mechanisms that get to see at most {\em $t$ samples} from this target distribution.

To overcome this, we consider the following modified variant of our adversary, called $\hBBBspace$. The modification is that $\hBBBspace$ does not get to see the full description of $\DDDsamp$. Instead it only gets to see points from the support of $\DDDsamp$ that correspond to indices in the set $I=[n]\setminus J$. Then, when generating the ciphertexts $c_j$, the modified adversary $\hBBBspace$ encrypts zeroes instead of points $m_j$ which it is missing. By the security of our MILR scheme, the mechanism $\AAAspace$ cannot notice this modification, and hence, assuming that it maintains accuracy against our original adversary $\BBBspace$ then it also maintains accuracy against our modified adversary $\hBBBspace$. As before, this means that $\hBBBspace$ maintains accuracy against the emulated $\BBBsamp$. Intuitively, this leads to a contradiction, as $\hBBBspace$ is using only $\overline{\tau}\leq t$ points from the target distribution $\DDDsamp$.

We stress that this proof sketch is over-simplified and inaccurate. In particular, the following two technical issues need to be addressed: (1) It is true that $\hBBBspace$ uses only $\overline{\tau}$ points from the support of the target distribution $\DDDsamp$, but these points are not necessarily {\em sampled} from $\DDDsamp$; and (2) The modified adversary $\hBBBspace$ is not computationally efficient because computing the set $J$ is not efficient. We address these issues, and other informalities made herein, in Appendix~\ref{app:ADAfullProof}.

\section{Construction of an MILR Scheme from a Semantically Secure Encryption Scheme}
\label{sec:constructME}

\paragraph{Construction.}
Let $\lambda' \leq \lambda$ be such that $\lambda = \poly(\lambda')$.
Given an encryption scheme $\Pi' = (\Gen',\Enc',\Dec')$ such that $\Gen'(1^{\lambda'})$ outputs a key uniformly distributed on $\{0,1\}^{\lambda'}$ (i.e., $x' \leftarrow_R \{0,1\}^{\lambda'}$),
we construct an MILR scheme $\Pi = (\Gen,\Param,\Enc,\Dec)$ as follows:
\begin{itemize}
    \item $\Gen$: On input $1^\lambda$, return $x \leftarrow_R \{0,1\}^{\lambda}$.
    \item $\Param$: On input $1^{\lambda}$, let $\mathcal{G}$ be a family of universal hash functions with domain $\{0,1\}^{\lambda}$ and range $\{0,1\}^{\lambda'}$.
        Return (a description of) $g \leftarrow_R \mathcal{G}$.
    \item $\Enc$: On input $(x,p,m)$, parse $g \coloneqq p$ (as a description of a hash function), let $x' = g(x)$ and return $\Enc'(x',m)$.
    \item $\Dec$: On input $(x,p,c)$, parse $g \coloneqq p$, let $x' = g(x)$ and return $\Dec'(x',c)$.
\end{itemize}

Using a standard construction of a universal hash function family, all the algorithms run in time polynomial in $\lambda$. Moreover, if $c\leftarrow\Enc(x,p,m)$, then $\Dec(x,p,c) = m$ with probability 1 (as this holds for $\Enc'$ and $\Dec'$).

The following two theorems (corresponding to the two security properties in Definition~\ref{def:muiltienc}) establish the security of $\Pi$ and prove Theorem~\ref{thm:multi_first_statement}.
\begin{theorem}[Multi semantic security]
\label{thm:multi1}
Let $\Omega(\lambda) \leq \Gamma(\lambda) \leq 2^{o(\lambda)}$ and $\lambda' = 0.1\lambda$.
If $\Pi'$ is $\Gamma(\lambda')$-secure against uniform (resp. non-uniform) adversaries, then $\Pi$ is $\Gamma(\lambda)$-secure against uniform (resp. non-uniform) adversaries.
\end{theorem}

\begin{theorem}[Multi-security against bounded preprocessing adversaries]
\label{thm:multi2}
Let $\Omega(\lambda) \leq \Gamma(\lambda) \leq 2^{o(\lambda)}$, $\lambda' = 0.1\lambda$ (as in Theorem~\ref{thm:multi1}).
If $\Pi'$ is $\Gamma(\lambda')$-secure  against non-uniform adversaries then $\Pi$ is $(\Gamma(\lambda),\overline{\tau})$-\emph{secure} against space bounded non-uniform preprocessing adversaries, where $\overline{\tau}(\lambda,s) = \frac{2s}{\lambda} + 4$.
\end{theorem}

\begin{remark}
Since $\Gamma(\lambda) \leq 2^{o(\lambda)}$ and $\lambda' = 0.1\lambda$, then
$\poly(\Gamma(\lambda')) = \poly(\Gamma(\lambda))$. Therefore, for the sake of simplicity, we analyze the runtime and advantage of all adversaries (including those that run against $\Pi'$) as functions of $\lambda$.
\end{remark}

We prove Theorem~\ref{thm:multi1} in Appendix~\ref{app:constructME} and Theorem~\ref{thm:multi2} in Section~\ref{sec:preprocessing}.

\section{Multi-security Against a Bounded Preprocessing Adversary}
\label{sec:preprocessing}

In this section we prove Theorem~\ref{thm:multi2}.
The proof requires specific definitions and notation, defined below.

\subsection{Preliminaries}

\paragraph{Notation.}

Given a sequence of elements $X = (X_1,\ldots,X_n)$  and a subset $I \subseteq [n]$, we denote by $X_I$ the sequence composed of elements with coordinates in $I$.

For a random variable $X$, denote its min-entropy by $H_{\infty}(X)$.
For random variables $X,Y$ with the same range, denote by $\Delta(X,Y)$ the statistical distance of their distributions. We say that $X$ and $Y$ are $\gamma$-close if
$\Delta(X,Y) \leq \gamma$.
We use the notation $X \leftarrow_R \mathcal{X}$ to indicate that the random variable $X$ is chosen uniformly at random from the set $\mathcal{X}$.

\paragraph{Dense and bit-fixing sources.}
We will use the following definition (see \cite[Definition 1]{CorettiDGS18}).
\begin{definition}
\label{def:source}
An $(n,2^{\lambda})$-source is a random variable $X$ with range $(\{0,1\}^\lambda)^n$. A source is called
\begin{itemize}
  \item $(1 - \delta)$-dense if for every subset $I \subseteq [n]$,
  $H_\infty(X_I) \geq (1 - \delta) \cdot |I| \cdot \lambda,$
  \item $(k, 1-\delta)$-dense if it is fixed on at most $k$ coordinates and is $(1 - \delta)$-dense on the rest,
  \item $k$-bit-fixing if it is fixed on at most $k$ coordinates and uniform on the rest.
\end{itemize}
\end{definition}
Namely, the min-entropy of every subset of entries of a $(1-\delta)$-dense source is at most a
fraction of $\delta$ less than what it would be for a uniformly random one.
\subsection{Key leakage lemma}

Let $X = (X_1,\ldots,X_n) \in (\{0,1\}^{\lambda})^n$ be a random variable for $n$ keys of $\Pi$ chosen independently and uniformly at random. Let $Z \coloneqq F(X)$ be a random variable for the leakage of the adversary. For $z \in \{0,1\}^s$, let $X_z$ be the random variable chosen from the distribution of $X$ conditioned on $F(X) = z$.

We denote $G \coloneqq (G_1,\ldots,G_n)$ and $G(X) \coloneqq (G_1(X_1), \ldots,G_n(X_n))$ the random variable for the hash functions (public parameters) of $\Pi$.
We will use similar notation for sequences of different lengths (which will be clear from the context).

The proof of Theorem~\ref{thm:multi2} is based on the lemma below (proved in Section~\ref{sec:leakage}), which analyzes the joint distribution $(G,Z,G(X))$.
\begin{lemma}
\label{lem:leakage}
Let $F: (\{0,1\}^{\lambda})^n \rightarrow \{0,1\}^s$ be an arbitrary function,
$X = (X_1,\ldots,X_n) \leftarrow_R (\{0,1\}^{\lambda})^n$ and denote $Z \coloneqq F(X)$.
Let $\mathcal{G}$ be a family of universal hash functions with domain $\{0,1\}^\lambda$ and range $\{0,1\}^{\lambda'}$ and let $G \leftarrow_R (\mathcal{G})^n$.
Let $\delta >0, \gamma > 0, s' > s$ be parameters such that $(1 - \delta) \lambda > \lambda' + \log n + 1$.

Then, there exists a family $V_{G,Z} = \{V_{\vec{g},z}\}_{\vec{g} \in (\mathcal{G})^n, z \in \{0,1\}^s}$ of convex combinations $V_{\vec{g},z}$ of $k$-bit-fixing $(n,2^{\lambda'})$-sources
for $k = \frac{s' + \log 1/\gamma}{\delta \cdot \lambda}$ such that
$$\Delta[(G,Z,G(X)),(G,Z,V_{G,Z})] \leq \sqrt{2^{-(1 - \delta) \lambda + \lambda' + \log n}} + \gamma + 2^{s - s'}.$$
\end{lemma}

We obtain the following corollary (which implies the parameters of Theorem~\ref{thm:multi2}).
\begin{corollary}
\label{cor:leakage}
In the setting of Lemma~\ref{lem:leakage}, assuming $n < 2^{0.15 \lambda}$ and sufficiently large $\lambda$, the parameters
$s' = s + \lambda,\lambda' =  0.1 \lambda, \delta = 0.5, \gamma = 2^{-\lambda}$ give
\begin{align*}
\Delta[(G,Z,G(X)),(G,Z,V_{G,Z})] \leq 2^{-0.1\lambda}, \text{ \qquad and \qquad } k = \frac{2s}{\lambda} + 4.
\end{align*}

\end{corollary}

\begin{proof}
Set $s' = s + \lambda$, $\lambda' =  0.1 \lambda$, $\delta = 0.5, \gamma = 2^{-\lambda}$.
Then, for sufficiently large $\lambda$,
$(1 - \delta) \lambda  = 0.5 \lambda > 0.1 \lambda + 0.15 \lambda + 1 > \lambda' + \log n + 1$
(and the condition of Lemma~\ref{lem:leakage} holds).
We therefore have (for sufficiently large $\lambda$): 
$
\Delta[(G,Z,G(X)),(G,Z,V_{G,Z})] \leq %
\sqrt{2^{-(1 - \delta) \lambda + \lambda' + \log n}} + \gamma + 2^{s - s'} =
\sqrt{2^{-0.4\lambda + \log n}} + 2^{-\lambda} + 2^{-\lambda} \leq 2^{-0.1\lambda},
$ 
and 
$
k = \frac{s' + \log 1/\gamma}{\delta \cdot \lambda} = \frac{s + \lambda + \lambda}{0.5 \cdot \lambda} = \frac{2s}{\lambda} + 4
$.

\end{proof}

\subsection{The Proof of Theorem~\ref{thm:multi2}}

\paragraph{Using Lemma~\ref{lem:leakage} to prove Theorem~\ref{thm:multi2}.}
Before proving Theorem~\ref{thm:multi2}, we explain why Lemma~\ref{lem:leakage} is needed and how it used in the proof.

It is easy to prove some weaker statements than Lemma~\ref{lem:leakage}, but these do not seem to be sufficient for building the MILR scheme (i.e., proving Theorem~\ref{thm:multi2}). For example, one can easily prove that with high probability, given the leakage $z$ and hash functions $\vec{g}$, there is a large subset of (hashed) keys such that each one of them is almost uniformly distributed. However, the adversary could have knowledge of various relations between the keys of this subset and it is not clear how to prove security without making assumptions about the resistance of the encryption scheme against related-key attacks.

Moreover, consider a stronger statement, which asserts that with high probability, given the leakage $z$ and hash functions $\vec{g}$, there is a large subset of (hashed) keys that are jointly uniformly distributed. We claim that even this stronger statement may not be sufficient to prove security, since it does not consider the remaining keys outside of the subset. In particular, consider a scenario in which the adversary is able to recover some weak keys outside of the subset. Given this extra knowledge and the leakage $z$, the original subset of keys may no longer be distributed uniformly (and may suffer from a significant entropy loss).

Lemma~\ref{lem:leakage} essentially asserts that there is a subset of keys that is almost jointly uniformly distributed even if we give the adversary $z,\vec{g}$ and all the remaining keys.
More specifically, given the hash functions $\vec{g}$ and the leakage $z$, according to the lemma,
the distribution of the hashed keys $\vec{g}(X)$ is (close to) a convex combinations $V_{\vec{g},z}$ of $k$-bit-fixing sources. In the proof of Theorem~\ref{thm:multi2} we will fix such a $k$-bit-fixing source by giving the adversary $k$ hashed keys (we will do this carefully, making sure that the adversary's advantage does not change significantly). Since the remaining hashed keys are uniformly distributed from the adversary's view, security with respect to these keys follows from the semantic security of the underlying encryption scheme.

\begin{proof}[Proof of Theorem~\ref{thm:multi2}]
Fix a preprecessing procedure $F$ and let $\lambda$ be sufficiently large, $n < 2^{0.15 \lambda}$.
By Lemma~\ref{lem:leakage} (with parameters set in Corollary~\ref{cor:leakage}),
there exists a family $V_{G,Z}$ of convex combinations $V_{\vec{g},z}$ of $k$-bit-fixing $(n,2^{\lambda'})$-sources
for
$k = \frac{2s}{\lambda} + 4$
such that
$\Delta[(G,Z,G(X)),(G,Z,V_{G,Z})] \leq 2^{-0.1\lambda}.$

\paragraph{Sampling the index set $J$ and simplifying the distribution.}
We first define how the oracles $\oracle_0$ and $\oracle_1$ in Definition~\ref{def:muiltienc} sample $J$. The random variable $J$ is naturally defined when sampling the random variables $(G,Z,V_{G,Z})$: given $\vec{g} \in (\mathcal{G})^n, z \in \{0,1\}^s$, sample a $k$-bit-fixing source in the convex combination $V_{\vec{g},z}$ (according to its weight) and let $J$ be the set of (at least) $n-k$ indices that are not fixed. This defines a joint distribution on the random variables $(G,Z,V_{G,Z},J)$. Another way to sample from this distribution is to first sample the variables $(G,Z,V_{G,Z})$ and then sample $J$ according to its marginal distribution. This defines a randomized procedure for sampling $J$. Although the oracles do not sample $(G,Z,V_{G,Z})$, we reuse the same sampling procedure for sampling $J$ given the sample $(\vec{g},z,\vec{g}(\vec{x}))$ (if the sample $(\vec{g},z,\vec{g}(\vec{x}))$ is not in the support of the distribution of $(G,Z,V_{G,Z})$, define $J = [n]$).

Consider a $\poly(\Gamma(\lambda))$-time algorithm $\BBB$.
As encryption queries of $\BBB$ to $\oracle_0$ and $\oracle_1$ are answered with the hash keys $\vec{g}(\vec{x})$, then given $\vec{g}(\vec{x})$, the interaction of $\BBB$ with $\oracle_0$ and $\oracle_1$ no longer depends on $\vec{x}$.
Therefore, for $t = 0,1$ we define $\oracle^{(1)}_t(\vec{g}(\vec{x}),\vec{g},J,\cdot,\cdot)$ that simulates the interaction of $\oracle_t(\vec{x},\vec{g},J,\cdot,\cdot)$ with $\BBB$. Instead of sampling $\vec{x}$, the oracles directly sample $(\vec{g},z,\vec{g}(\vec{x}),J)$ according to their joint distribution before the interaction with $\BBB$. To simplify notation, we denote this joint distribution by $\mathcal{D}_1$.
Denote
\begin{align*}
&\adv_{\BBB}(\lambda) = \\
&\left|
\Pr_{\substack{\BBB,\Enc\\
(\vec{g},z,\vec{g}(\vec{x}),J)
\leftarrow \mathcal{D}_1}}
\left[  \BBB^{\oracle^{(1)}_0(\vec{g}(\vec{x}),\vec{g},J,\cdot,\cdot)}(z,\vec{g})=1  \right]
-
\Pr_{\substack{\BBB,\Enc\\
(\vec{g},z,\vec{g}(\vec{x}),J)
\leftarrow \mathcal{D}_1}}\left[  \BBB^{\oracle^{(1)}_1(\vec{g}(\vec{x}),\vec{g},J,\cdot,\cdot)}(z,\vec{g})=1  \right]
\right|.
\end{align*}
It remains to prove that $\adv_{\BBB}(\lambda) \leq \negl(\Gamma(\lambda))$.

\paragraph{Using Lemma~\ref{lem:leakage} to switch to a family of convex combinations of bit-fixing sources.}
We have
$$\Delta[(G,Z,G(X),J),(G,Z,V_{G,Z},J)] \leq
\Delta[(G,Z,G(X)),(G,Z,V_{G,Z})] \leq 2^{-0.1\lambda},
$$
where the first inequality follows by the data processing inequality, since $J$ is computed by applying the same function to the three variables of both distributions, and the second inequality is by Corollary~\ref{cor:leakage}.
Hence, for $t= 0,1$ we replace $\oracle^{(1)}_t$ that samples from $\mathcal{D}_1$ with $\oracle^{(2)}_t$ that samples from the joint distribution of $(G,Z,V_{G,Z},J)$, which we denote by $\mathcal{D}_2$. Since $\BBB$ and $\Enc$ use independent randomness, by the triangle inequality, the total penalty is at most $2 \cdot 2^{-0.1\lambda}$, namely
\begin{align*}
&\left|
\Pr_{\substack{\BBB,\Enc\\
(\vec{g},z,\vec{y},J)
\leftarrow \mathcal{D}_2}}
\left[  \BBB^{\oracle^{(2)}_0(\vec{y},\vec{g},J,\cdot,\cdot)}(z,\vec{g})=1  \right]
- 
\Pr_{\substack{\BBB,\Enc\\
(\vec{g},z,\vec{y},J)
\leftarrow \mathcal{D}_2}}\left[  \BBB^{\oracle^{(2)}_1(\vec{y},\vec{g},J,\cdot,\cdot)}(z,\vec{g})=1  \right]
\right| \geq \\
&\adv_{\BBB}(\lambda) - 2^{-0.1\lambda + 1},
\end{align*}
where we denote a sample from $\mathcal{D}_2$ by $(\vec{g},z,\vec{y},J)$.

\paragraph{Giving the adversary additional input.}
Consider a (potentially) more powerful $\poly(\Gamma(\lambda))$-time algorithm $\BBB_1$ against $\Pi$ whose input consists of $(z,\vec{g},J,\vec{y}_{\overline{J}})$, where $\overline{J} = [n] \backslash J$.
Namely, in addition to $(z,\vec{g})$ the input also consists of $J$, as well as the hashed keys $\vec{y}_{\overline{J}} \in \left(\{0,1\}^{\lambda'}\right)^{n - |J|}$ (note that these parameters define a $|\overline{J}|$-bit-fixing source). We denote $in = (z,\vec{g},J,\vec{y}_{\overline{J}})$,
and
\begin{align*}
&\adv_{\BBB_1}(\lambda) = \\
&\left|
\Pr_{\substack{\BBB_1,\Enc\\
(\vec{g},z,\vec{y},J)
\leftarrow \mathcal{D}_2}}
\left[  \BBB_1^{\oracle^{(2)}_0(\vec{y},\vec{g},J,\cdot,\cdot)}(in)=1  \right]
-
\Pr_{\substack{\BBB_1,\Enc\\
(\vec{g},z,\vec{y},J)
\leftarrow \mathcal{D}_2}}\left[  \BBB_1^{\oracle^{(2)}_1(\vec{y},\vec{g},J,\cdot,\cdot)}(in)=1  \right]
\right|.
\end{align*}
Next, we prove that for any such $\poly(\Gamma(\lambda))$-time algorithm $\BBB_1$, $\adv_{\BBB_1}(\lambda) \leq \negl(\Gamma(\lambda))$.
As any algorithm $\BBB$ with input $(z,\vec{g})$ can be simulated by an algorithm $\BBB_1$ with input $in$ and similar runtime, this implies that $\adv_{\BBB}(\lambda) - 2^{-0.1\lambda + 1} \leq \negl(\Gamma(\lambda))$ and hence $\adv_{\BBB}(\lambda) \leq \negl(\Gamma(\lambda))$, concluding the proof.

\paragraph{Fixing the adversary's input.}
Since both $\oracle^{(2)}_0$ and $\oracle^{(2)}_1$ sample the input of $\BBB_1$ from the same distribution, by an averaging argument, there exists an input $in^{*} = in^{*}_{\lambda} =   (z^{*},\vec{g}^{*},J^{*},\vec{y^{*}}_{\overline{J^{*}}})$ such that the advantage of $\BBB_1$ remains at least as large when fixing the input to $in^{*}$ and sampling from $\mathcal{D}_2$ conditioned on $in^{*}$. Note that given $in^{*}$ sampling from $\mathcal{D}_2$ reduces to sampling from the $|\overline{J^{*}}|$-bit-fixing source defined by $(\overline{J^{*}},\vec{y^{*}}_{\overline{J^{*}}})$,
i.e., selecting $\vec{w} \leftarrow_R (\{0,1\}^{\lambda'})^{|J^{*}|}$.
Therefore,
\begin{align*}
&\left|
\Pr_{\substack{\BBB_1,\Enc\\
\vec{w} \leftarrow_R (\{0,1\}^{\lambda'})^{|J^{*}|}}}
\left[  \BBB_1^{\oracle^{(2)}_0(in^{*},\vec{w},\cdot,\cdot)}(in^{*})=1  \right]
-
\Pr_{\substack{\BBB_1,\Enc\\
\vec{w} \leftarrow_R (\{0,1\}^{\lambda'})^{|J^{*}|}}}\left[  \BBB_1^{\oracle^{(2)}_1(in^{*},\vec{w},\cdot,\cdot)}(in^{*})=1  \right]
\right| \geq \\
&\adv_{\BBB_1}(\lambda).
\end{align*}

\paragraph{Reducing the security of $\Pi$ with preprocessing from the (multi-instance) security of $\Pi'$.}
We now use $\BBB_1$ to define a non-uniform $\poly(\Gamma(\lambda))$-time adversary $\BBB_2$ (with no preprocessing) that runs against $|J^{*}|$ instances of $\Pi'$ and has advantage at least $\adv_{\BBB_1}(\lambda)$.
By the semantic security of $\Pi'$ and a hybrid argument (similarly to the proof of Theorem~\ref{thm:multi1}), this implies that $\adv_{\BBB_1}(\lambda) \leq \negl(\Gamma(\lambda))$, concluding the proof.

\begin{algorithm}\caption{$\BBB_2^{\oracle_{(\cdot)}^{(3)}(\vec{x'},[|J^{*}|],\cdot,\cdot)}()$}
\label{alg:b2_2}

{\bf Setting:} $\BBB_2$ is a non-uniform adversary that runs against $|J^{*}|$ instances of $\Pi'$ (defined by $\oracle_{(\cdot)}^{(3)}$).
It gets $in^{*} = in^{*}_{\lambda} = (z^{*},\vec{g}^{*},J^{*},\vec{y^{*}}_{\overline{J^{*}}})$ as advice.
$\BBB_2$ has access to $\BBB_1^{\oracle^{(2)}_{(\cdot)}(in^{*},\vec{w},\cdot,\cdot)}(in^{*})$, which runs against $\Pi$.

\begin{enumerate}
  \item $\BBB_2$ gives $in^{*}$ to $\BBB_1$ as input.
  \item $\BBB_2$ answers each query $(j,m)$ of $\BBB_1$ as follows:
  \begin{itemize}
    \item If $j \in \overline{J^{*}}$, $\BBB_2$ uses the advice string $in^{*}$ (which contains $\vec{y}^{*}_j$) to compute the answer $\Enc'(\vec{y}^{*}_j,m)$ and gives it to $\BBB_1$.
    \item If $j \in J^{*}$, $\BBB_2$ translates the query $(j,m)$ to $(j',m)$, where $j' \in [|J^{*}|]$ is obtained by mapping $j$ to $J^{*}$ (ignoring indices in $\overline{J^{*}}$). $\BBB_2$ then queries its oracle with $(j',m)$ and forwards the answer to $\BBB_1$.
  \end{itemize}
  \item $\BBB_2$ outputs the same output as $\BBB_1$.
\end{enumerate}

\end{algorithm}

The adversary $\BBB_2$ is given in Algorithm~\ref{alg:b2_2}.
Note that $\BBB_2$ perfectly simulates the oracles of $\BBB_1$ given the input $in^{*}$, and hence its advantage is at least $\adv_{\BBB_1}(\lambda)$ as claimed. Finally, it runs in time $\poly(\Gamma(\lambda))$.
\end{proof}

\subsection{Proof of Lemma~\ref{lem:leakage}}
\label{sec:leakage}

\paragraph{Proof overview.}
We first prove in Lemma~\ref{lem:leftover} that $G = (G_1,\ldots,G_t)$ (for some $t \in [n]$) is a good extractor, assuming its input $Y = (Y_1,\ldots,Y_t)$ is $(1 - \delta)$-dense, namely, it has sufficient min-entropy for each subset of coordinates (see Lemma~\ref{lem:leftover} for the exact statement). Specifically, we prove that $(G,G(Y))$ is statistically close to $(G,U)$, where $U$ is uniformly distributed over $(\{0,1\}^{\lambda'})^t$. The proof is by a variant of the leftover hash lemma~\cite{HastadILL99} where a sequence of hash functions $(G_1,\ldots,G_t)$ are applied locally to each block of the input (instead of applying a single hash function to the entire input). We note that a related lemma was proved in~\cite[Lem. 13]{GoosLM0Z15} in a different setting of communication complexity.
Our variant is applicable to a different (mostly wider) range of parameters (such as various values of $\delta$ and the number of bits extracted, $t \cdot \lambda'$) that is relevant in our setting. Additional (somewhat less related) results were presented in~\cite{DodisFMT20,TZ20}.

The remainder of the proof is given in Appendix~\ref{app:preprocessing} and is somewhat similar to~\cite[Lem. 1]{CorettiDGS18}.
We fix $z \in \{0,1\}^s$ such that $H_{\infty}(X_z)$ is sufficiently high.
Lemma~\ref{lem:fixing} (proved in~\cite{CorettiDGS18,GoosLM0Z15,KothariMR17}) states that $X_z$ (before applying $G$) is close to a convex combination of $(k,1 - \delta)$-dense sources (for a sufficiently small value of $k$).
We therefore invoke Lemma~\ref{lem:leftover} on the high min-entropy coordinates of each such $(1-\delta)$-dense source. This allows to derive Lemma~\ref{lem:zleakage}, which states that $(G,G(X_z))$ is statistically close to $(G,V_{G,z}^{*})$, where $V_{G,z}^{*}$ is a family of convex combinations of $k$-bit-fixing $(n,2^{\lambda'})$-sources, indexed by $\vec{g} \in \mathcal{G}^n$.

Proposition~\ref{prop:min} (which is quite standard) shows that with very high probability (over $X$ and $F$, setting $z \leftarrow F(X)$), $H_{\infty}(X_z)$ is indeed sufficiently high to apply Lemma~\ref{lem:zleakage} with good parameters.
For each such $z$, we define $V_{G,z}$ to be $V_{G,z}^{*}$ as above and this essentially results in the family $V_{G,Z}$ (the remaining values of $z$ with low $H_{\infty}(X_z)$ are easy to handle with small loss in statistical distance).
Thus, the statistical distance  between $(G,Z,G(X))$ and $(G,Z,V_{G,Z})$ can be bounded by combining Proposition~\ref{prop:min} and Lemma~\ref{lem:zleakage}, establishing Lemma~\ref{lem:leakage}.

\paragraph{Block-wise extraction from dense sources.}
\begin{lemma}
\label{lem:leftover}
Let $Y = (Y_1,\ldots,Y_t) \in (\{0,1\}^{\lambda})^t$ be a $(t,2^{\lambda})$-source that is $(1 - \delta)$-dense for $0< \delta < 1$.
Let $\mathcal{G}$ be a family of universal hash functions with domain $\{0,1\}^\lambda$ and range $\{0,1\}^{\lambda'}$. Then, for $G \leftarrow_R (\mathcal{G})^t$ and $U \leftarrow_R (\{0,1\}^{\lambda'})^t$,
\begin{align*}
\Delta[(G,G(Y)),(G,U)] \leq \sqrt{2^{-(1 - \delta) \lambda + \lambda' + \log t}},
\end{align*}
assuming that $(1 - \delta) \lambda > \lambda' + \log t + 1$.
\end{lemma}

\begin{proof} Let $d \coloneqq \log |\mathcal{G}|$. For a random variable $Q$, and $Q'$ an independent copy of $Q$, we denote by $Col[Q] = \Pr[Q = Q']$ the collision probability of $Q$. We have
\begin{align}
\label{eq:col1}
\begin{split}
&Col[(G,G(Y))] =
\Pr_{G,Y,G',Y'}[(G,G(Y')) = (G',G'(Y'))] = \\
&\Pr_{G,G'}[G = G'] \cdot \Pr_{G,Y,Y'}[G(Y) = G(Y')] =
2^{-t \cdot d} \cdot  \Pr_{G,Y,Y'}[G(Y) = G(Y')].
\end{split}
\end{align}
For sequences $Y_1,\ldots,Y_t,Y'_1,\ldots,Y'_t$, define $C = |\{i \mid Y_i = Y'_i\}|$. We now upper bound the expression $\Pr[C = c]$.

Recall that $Y$ is a $(1 - \delta)$-dense source, i.e., for every subset $I \subseteq [t]$,
$H_\infty(Y_I) \geq (1 - \delta) \cdot |I| \cdot \lambda$. Fix a subset $I \subseteq [t]$ such that $|I| = c$. Then,
\begin{align*}
&\Pr[Y_I = Y'_I] = \sum_{y_I \in (\{0,1\}^\lambda)^{c}} (\Pr[Y_I = y_I])^2 \leq \\
&\max_{y_I}\{\Pr[Y_I = y_I]\} \cdot \sum_{y_I \in (\{0,1\}^\lambda)^{c}} \Pr[Y_I = y_I] \leq
2^{-(1 - \delta) \cdot c \cdot \lambda}.
\end{align*}
Therefore,
\begin{align}
\label{eq:col2}
\begin{split}
&\Pr[C = c] \leq
\sum_{\{I \subseteq [t] \mid |I| = c\}} \Pr[Y_I = Y'_I] \leq
\binom{t}{c} \cdot 2^{-(1 - \delta) \cdot c \cdot \lambda} \leq \\
&t^c \cdot 2^{-(1 - \delta) \cdot c \cdot \lambda} =
2^{c \cdot (-(1 - \delta) \lambda + \log t)}.
\end{split}
\end{align}

We have
\begin{align*}
\Pr_{G,Y,Y'}[G(Y) = G(Y')] =
\sum_{c = 0}^{t} \Pr_{Y,Y'}[C = c] \cdot \Pr_{G,Y,Y'}[G(Y) = G(Y') \mid C = c].
\end{align*}

For each coordinate $i$ such that $Y_i \neq Y'_i$,
$\Pr_{G_i}(G_i(Y_i) = G_i(Y'_i)) = 2^{-\lambda'}$ as $G_i$ is selected uniformly from a family of universal hash functions. Since $G = (G_1,\ldots,G_t)$ contains $t$ independent copies selected uniformly from $\mathcal{G}$,
\begin{align*}
\Pr_{G}[G(Y) = G(Y') \mid C = c] =
2^{-\lambda' \cdot (t-c)}.
\end{align*}

Hence, using~(\ref{eq:col2}) we obtain
\begin{align*}
&\Pr_{G,Y,Y'}[G(Y) = G(Y')] =
\sum_{c = 0}^{t} \Pr[C = c] \cdot 2^{-\lambda' \cdot (t-c)} \leq \\
&\sum_{c = 0}^{t} 2^{c \cdot (-(1 - \delta) \lambda + \log t)} \cdot 2^{-\lambda' \cdot (t-c)} =
2^{- \lambda' \cdot t} \cdot \sum_{c = 0}^{t} 2^{- c \cdot ((1 - \delta) \lambda - \lambda' -  \log t)} = \\
&2^{- \lambda' \cdot t} \cdot (1 + \sum_{c = 1}^{t} 2^{- c \cdot ((1 - \delta) \lambda - \lambda' - \log t)}) \leq
2^{- \lambda' \cdot t} \cdot (1 + 2^{-(1 - \delta) \lambda + \lambda' + \log t + 1}),
\end{align*}
where the last inequality uses the assumption that $(1 - \delta) \lambda > \lambda' + \log t + 1$.

Treating distributions as vectors over $\{0,1\}^{t \cdot d + t \cdot \lambda'}$ (and abusing notation), we plug the above expression into~(\ref{eq:col1}) and deduce
\begin{align*}
&\|(G,G(Y)) - (G,U) \|_2^2 = Col[(G,G(Y))] - 2^{-t \cdot d - t \cdot \lambda'} \leq \\
&2^{-t \cdot d - t \cdot \lambda'} \cdot (1 +  2^{-(1 - \delta)\lambda + \lambda' + \log t + 1}) - 2^{-t \cdot d - t \cdot \lambda'} =
2^{-t \cdot d - t \cdot n' -(1 - \delta) \lambda + \lambda' + \log t + 1}.
\end{align*}

Finally, using the Cauchy–Schwarz inequality, we conclude
\begin{align*}
&\Delta[(G,G(Y)),(G,U)] \leq
1/2 \cdot \sqrt{2^{t \cdot d + t \cdot \lambda'}} \cdot \|(G,G(Y)) - (G,U) \|_2 \leq  \\
&1/2 \cdot \sqrt{2^{t \cdot d + t \cdot \lambda'}} \cdot \sqrt{ 2^{-t \cdot d - \lambda' \cdot t -(1 - \delta) \lambda + \lambda' + \log t + 1} } <
\sqrt{2^{-(1 - \delta) \lambda + \lambda' + \log t}}.
\end{align*}

\end{proof}

%


\appendix

\section{Additional preliminaries from cryptography}\label{sec:additionalPrelims}

A {\em probability ensemble} is an infinite sequence of probability distributions. Computational indistinguishability is defined as follows.

\begin{definition}
Two probability ensembles $\XXX=\{X_{\lambda}\}_{\lambda\in\N}$ and $\YYY=\{Y_{\lambda}\}_{\lambda\in\N}$ are {\em computationally indistinguishable}, denoted as $\XXX\equiv_c\YYY$, if for
every non-uniform probabilistic polynomial-time distinguisher $D$ there exists a negligible
function $\negl$ such that:
$$
\left|
\Pr_{x\leftarrow X_{\lambda}}\left[D(x)=1\right]
-
\Pr_{y\leftarrow Y_{\lambda}}\left[D(y)=1\right]
\right|
\leq\negl(\lambda).
$$
\end{definition}

\begin{definition}[\cite{Goldreich2001}]
A {\em non-uniform polynomial-time algorithm} is a pair $(M, \overline{a})$, where $M$ is a two-input polynomial-time Turing machine and $\overline{a} = a_1, a_2,\dots$ is an infinite sequence of strings
such that $|a_{\lambda}| = \poly(\lambda)$. For every $x$, we consider the computation of $M$
on the input pair $(x, a_{|x|})$. Intuitively, $a_{\lambda}$ can be thought of as extra ``advice'' supplied
from the ``outside'' (together with the input $x\in\{0, 1\}^{\lambda}$). We stress that $M$ gets the same advice (i.e., $a_{\lambda}$) on all inputs of the same length (i.e., $\lambda$). 

Equivalently, a {\em non-uniform polynomial-time algorithm} is an infinite sequence of Turing machines, $M_1, M_2,\dots$, such that both the length of the description of $M_{\lambda}$ and its running time on inputs of length $\lambda$ are bounded by polynomials
in $\lambda$ (fixed for the entire sequence). Machine $M_{\lambda}$ is used only on inputs of length $\lambda$.
\end{definition}

\paragraph{Encryption schemes.}
An {\em encryption scheme} is a tuple of efficient 
algorithms $(\Gen,\Enc,\Dec)$ with the following syntax:
\begin{itemize}
    \item $\Gen$ is a randomized algorithm that takes as input a security parameter $\lambda$ and outputs a $\lambda$-bit secret key. Formally,
    $x\leftarrow\Gen(1^\lambda)$.
    \item $\Enc$ is a randomized algorithm that takes as
    input a secret key $x$, a public parameter $p$, and a
    message $m\in\{0,1\}$, and outputs a
    ciphertext $c\in\{0,1\}^{\poly(\lambda)}$. Formally, $c\leftarrow\Enc(x,m)$.
    \item $\Dec$ is a deterministic algorithm that takes as input a secret key $x$, a public parameter $p$, and a ciphertext $c$, and outputs a decrypted message $m'$. If the ciphertext $c$ was an encryption of $m$ under the key $x$ with the parameter $p$, then $m' = m$. Formally, if $c\leftarrow\Enc(x,m)$, then $\Dec(x,c) = m$ with probability $1$.
\end{itemize}

Let us recall the formal definition of security of an encryption scheme. Consider a pair of oracles $\oracle_0$ and $\oracle_1$, where $\oracle_1(x,\cdot)$ takes as input a message $m$ and returns $\Enc(m,x)$, and where $\oracle_0(x,\cdot)$ takes the same input
but returns $\Enc(0,x)$. 
An encryption scheme $(\Gen,\Enc,\Dec)$ is {\em secure} if no computationally efficient adversary can tell whether it is interacting with $\oracle_0$ or with $\oracle_1$. Formally,

\begin{definition}
Let $\lambda$ be a security parameter, and let $\Gamma=\Gamma(\lambda)$ be a function of $\lambda$. 
An encryption scheme $(\Gen,\Enc,\Dec)$ is $\Gamma$-\emph{secure} if for every  $\poly(\Gamma(\lambda))$-time adversary $\BBB$ there exists a negligible
function $\negl$ such that:
\begin{equation*}
\left| 
\Pr_{\substack{x\leftarrow\Gen(1^\lambda) \\ \BBB,\Enc }}\left[  \BBB^{\oracle_0(x,\cdot)}=1  \right] 
- 
\Pr_{\substack{x\leftarrow\Gen(1^\lambda) \\ \BBB,\Enc }}\left[  \BBB^{\oracle_1(x,\cdot)}=1  \right] 
\right| \leq \negl(\Gamma(\lambda)),
\end{equation*}
where the probabilities are over sampling $x\leftarrow\Gen(1^{\lambda})$ and over the randomness of $\BBB$ and $\Enc$.
\end{definition}

\section{Proof of Theorem~\ref{thm:ADA}}\label{app:ADAfullProof}

In this section we prove Theorem~\ref{thm:ADA}. We begin by recalling additional preliminaries from \cite{HardtU14,SteinkeU15}. As a stepping stone towards proving their negative results for the ADA problem, \cite{HardtU14,SteinkeU15} showed a negative result for a restricted family of mechanisms, called {\em natural} mechanisms. These are algorithms that can only evaluate the query on the sample points they are given. Formally,

\begin{definition}[\cite{HardtU14}]
A mechanism $\AAA$ is {\em natural} if, when given a sample $Y=(y_1,\dots,y_t)\in X^t$ 
and a query $q:X\rightarrow\{-1,0,1\}$, the algorithm returns an answer that is a function
only of $(q(y_1),\dots ,q(y_t))$. In particular, it cannot evaluate $q$ on other data points of its choice.
\end{definition}

Consider the following variant of the adaptive game that appeared in Figure~\ref{fig:ADAintro} (in the introduction), in which the underlying distribution is fixed to be uniform. See algorithm \texttt{AdaptiveGameOne}.

\begin{algorithm}\caption{$\texttt{AdaptiveGameOne}(\AAA,\BBB,t,k)$}
{\bf Notation:} $t$ is the sample size and $n=2000t$ is the domain size. 

\begin{enumerate}
    \item The mechanism $\AAA$ gets a sample $Y=(y_1,\dots,y_t)$ containing $t$ i.i.d.\ samples from the uniform distribution on $[n]$. This sample is not given to $\BBB$.
    \item For round $i=1,2,\dots,k$:
    \begin{enumerate}
        \item The adversary $\BBB$ specifies a query $q_i:[n]\rightarrow\left\{-1,0,1\right\}$
        \item The mechanism $\AAA$ obtains $q_i$  %
        and responds with an answer $a_i\in[-1,1]$ 
        \item[] \gray{{\small \% If $\AAA$ is {\em natural} then it only gets $(q_i(y_1),\dots,q_i(y_t))$ instead of the full description of $q_i$.}}
        \item $a_i$ is given to $\BBB$
    \end{enumerate}
    \item The outcome of the  game is one if $\exists i$ s.t.\ $\left|a_i - \frac{1}{n}\sum_{j\in[n]}q_i(j) \right|>1/10$, and zero otherwise.
\end{enumerate}
\end{algorithm}

\begin{theorem}[\cite{HardtU14,SteinkeU15}]
There exists a computationally efficient adversary $\BBB$ such that for every {\em natural} mechanism $\AAA$ and every large enough $t$ and $k=\Theta(t^2)$ it holds that
$$
\Pr[\texttt{AdaptiveGameOne}(\AAA,\BBB,t,k)=1]>\frac{3}{4}.
$$
\end{theorem}

We remark that this theorem holds  even if $M$ is computationally unbounded. That is, there exists a computationally efficient adversary $\AAA$ that fails {\em every} natural mechanism $M$ in the ADA problem, even when the underlying distribution is uniform.

Our first observation is that the proof of \cite{HardtU14,SteinkeU15} continues to hold even if in Step 1 of the game the mechanism {\em chooses} its $t$ input points, rather than sampling them uniformly. We refer to such a mechanism as a {\em selecting} mechanism. See algorithm \texttt{AdaptiveGameTwo}.

\begin{algorithm}\caption{$\texttt{AdaptiveGameTwo}(\AAA,\BBB,t,k)$}
{\bf Notations:} $t$ is the sample size and $n=2000t$ is the domain size. 

\begin{enumerate}
    \item \red{The mechanism $\AAA$ {\em chooses} a subset  $Y=(y_1,\dots,y_t)\subseteq[n]$ of size $t$. This subset is not given to $\BBB$.}
    \item For round $i=1,2,\dots,k$:
    \begin{enumerate}
        \item The adversary $\BBB$ specifies a query $q_i:[n]\rightarrow\left\{-1,0,1\right\}$
        \item The mechanism $\AAA$ obtains $q_i$  %
        and responds with an answer $a_i\in[-1,1]$ 
        \item[] \gray{{\small \% If $\AAA$ is {\em natural} then it only gets $(q_i(y_1),\dots,q_i(y_t))$ instead of the full description of $q_i$.}}
        \item $a_i$ is given to $\BBB$
    \end{enumerate}
    \item The outcome of the  game is one if $\exists i$ s.t.\ $\left|a_i - \frac{1}{n}\sum_{j\in[n]}q_i(j) \right|>1/10$, and zero otherwise.
\end{enumerate}
\end{algorithm}

\begin{theorem}[extension of \cite{HardtU14,SteinkeU15}]\label{thm:extension}
There exists a computationally efficient adversary $\BBB$ such that for every {\em natural} mechanism $\AAA$ and every large enough $t$ and $k=\Theta(t^2)$ it holds that
$$
\Pr[\texttt{AdaptiveGameTwo}(\AAA,\BBB,t,k)=1]>\frac{3}{4}.
$$
\end{theorem}

\begin{remark}
This extension is not directly necessary for our purposes, but it simplifies our analysis. Recall that we ultimately want to show a negative result for efficient mechanisms that can preprocess the underlying distribution and store a limited amount of space about it. Intuitively, using our notion of MILR, we will force such a mechanism to store useful information only about a bounded number of points from the underlying distribution. But these points are not necessarily {\em sampled} from the underlying distribution, and hence, it will be more convenient for us to rely on a hardness result for {\em selecting} natural mechanisms.
\end{remark}

\subsection{From selecting and natural to space boundedness}

We next ``lift'' the negative result for selecting and natural mechanisms (Theorem~\ref{thm:extension}) to obtain space a lower bound. To this end, recall \texttt{AdaptiveGameSpace}, presented in Section~\ref{sec:ADAbody}, where the mechanism initially gets the full description of the underlying distribution, but it must shrink it into an $s$-bit summary $z$. 

We prove the following theorem.

\begin{theorem}\label{thm:mainADAbody} 
Assume the existence of an MILR scheme $\Pi=(\Gen,\Param,\Enc,\Dec)$  with security parameter $\lambda$ that is  $(\Gamma,\overline{\tau})$-secure against space bounded preprocessing adversaries. Let $k\leq\poly(\Gamma(\lambda))$ be a function of $\lambda$. 
There exists a $\poly(\Gamma(\lambda))$-time adversary $\BBB$ 
such that the following holds. Let $\AAA{=}(\AAA_1,\AAA_2)$ be a $\poly(\Gamma(\lambda))$-time mechanism with space complexity $s$ such that 
$\overline{\tau}(\lambda,s)\leq O\left(\sqrt{k}\right)$ (small enough). Then,
$$
\Pr[\texttt{AdaptiveGameSpace}(\AAA,\BBB,s,k)=1]>\frac{2}{3}.
$$

Furthermore, the underlying distribution defined by the adversary $\BBB$ can be fully described using $O(\sqrt{k} \cdot \lambda)$ bits, 
it is sampleable in $\poly(\Gamma(\lambda))$-time, and elements sampled from this distribution can be represented using $O(\lambda+\log(k))$ bits.
\end{theorem}

\begin{proof}
Our task is to show that there is an adversary that fails every (efficient) mechanism that plays in \texttt{AdaptiveGameSpace}. What we know is that there is an adversary that fails every natural mechanism that plays in \texttt{AdaptiveGameTwo}. To bridge this gap 
we consider a wrapper algorithm, called \texttt{WrapA}, that can be used as a wrapper over a mechanism $\AAA$ designed for \texttt{AdaptiveGameSpace}, in order to transform it into a {\em natural} mechanism that plays in \texttt{AdaptiveGameTwo}. Consider algorithm \texttt{WrapA}.

We stress that algorithm $\texttt{WrapA}$ is not necessarily computationally efficient, since computing $J(\AAA_1,\vec{x},z,\vec{p})$ might not be efficient. This does not break our analysis because Theorem~\ref{thm:extension} holds even for computationally unbounded mechanisms (that are natural). Indeed, for every mechanism $\AAA$ (designed for \texttt{AdaptiveGameSpace}) we have that $\texttt{WrapA}{\circ}\AAA$ is a {\em natural} mechanism for playing in \texttt{AdaptiveGameTwo}. Therefore, the adversary $\BBB$ guaranteed by Theorem~\ref{thm:extension} causes $\texttt{WrapA}{\circ}\AAA$ to fail in \texttt{AdaptiveGameTwo}, provided that
\begin{equation}
\overline{\tau}(\lambda,s)=|I|\leq t = \Theta(\sqrt{k}),
\label{eq:ADAcondition}
\end{equation}
where $I$ is the set from Step~\ref{step:WrapAI} of algorithm \texttt{WrapA} and $t$ is the sample size in \texttt{AdaptiveGameTwo}. We continue with the analysis assuming that Inequality~(\ref{eq:ADAcondition}) holds. 

\begin{algorithm}\caption{$\texttt{WrapA}$}

{\bf Setting:} Algorithm \texttt{WrapA} has access to an algorithm $\AAA{=}(\AAA_1,\AAA_2)$ with space complexity $s$ (designed for \texttt{AdaptiveGameSpace}) and it interacts with an adversary $\BBB$ (designed for \texttt{AdaptiveGameTwo}).

\begin{enumerate}
    \item Let $(\Gen,\Param,\Enc,\Dec)$ be a $(\Gamma,\overline{\tau})$-secure MILR scheme with security parameter $\lambda$.
    \item Sample $n$ keys $x_1,\dots,x_n\in\{0,1\}^{\lambda}$ from $\Gen$, and define $\DDD$ to be the uniform distribution over the set $\{(j,x_j)\}_{j=1}^n$.
    \item Let $z\leftarrow \AAA_1(\DDD)$.
    
    \item Sample $n$ public parameters $p_1,\dots,p_n\in\{0,1\}^{\poly(\lambda)}$ from $\Param$.
    
    \item Let $J\leftarrow J(\AAA_1,\vec{x},z,\vec{p})\subseteq[n]$ be the set of ``hidden coordinates'', as in Definition~\ref{def:muiltienc}. Denote $I=[n]\setminus J$. The set $I$ corresponds to the subset $Y$ from Step 1 of \texttt{AdaptiveGameTwo}.\label{step:WrapAI}
    
    \item Instantiate $\AAA_2$ with $z$.
    
    \item For round $i=1,2,\dots,k$:
    \begin{enumerate}
        \item Obtain a query $q_i$, restricted to coordinates in $I$. That is, for every $j\in I$ we obtain $q_i(j)$. 
        \item For $j\in I$ let $c_{i,j}=\Enc(x_i,p_i,q_i(j))$. For $j\notin I$ let $c_{i,j}=\Enc(x_i,p_i,0)$.
        \item Define the query $f_i:[n]\times\{0,1\}^\lambda\rightarrow\{-1,0,1\}$ where $f_i(j,x)=\Dec(x,p_j,c_{i,j})$. The description of $f_i$ includes $\{c_{i,j}\}_{j\in[n]}$ and $\{p_j\}_{j\in[n]}$.
        \item Give $f_i$ to $\AAA_2$ and obtain an answer $a_i\in[-1,1]$.
        \item Output $a_i$.
    \end{enumerate}
\end{enumerate}
\end{algorithm}

Recall that our task is to design an adversary that causes $\AAA$ to fail in \texttt{AdaptiveGameSpace}. To this end, let us consider a modified version of algorithm \texttt{WrapA}, called \texttt{WrapB}. %

\begin{algorithm}\caption{$\texttt{WrapB}$}

{\bf Setting:} Algorithm \texttt{WrapB} has access to an algorithm $\AAA{=}(\AAA_1,\AAA_2)$ with space complexity $s$  (designed for \texttt{AdaptiveGameSpace}) and it interacts with an adversary $\BBB$ (designed for \texttt{AdaptiveGameTwo}).

\begin{enumerate}
    \item Let $(\Gen,\Param,\Enc,\Dec)$ be a $(\Gamma,\overline{\tau})$-secure MILR scheme with security parameter $\lambda$.
    \item Sample $n$ keys $x_1,\dots,x_n\in\{0,1\}^{\lambda}$ from $\Gen$, and define $\DDD$ to be the uniform distribution over the set $\{(j,x_j)\}_{j=1}^n$.
    \item Let $z\leftarrow \AAA_1(\DDD)$.

    \item Sample $n$ public parameters $p_1,\dots,p_n\in\{0,1\}^{\poly(\lambda)}$ from $\Param$.
    
    \item Let $J\leftarrow J(\AAA_1,\vec{k},z,\vec{p})\subseteq[n]$ be the set of ``hidden coordinates'', as in Definition~\ref{def:muiltienc}. Denote $I=[n]\setminus J$. The set $I$ corresponds to the subset $Y$ from Step 1 of \texttt{AdaptiveGameTwo}.
    
    \item Instantiate $\AAA_2$ with $z$.
    
    \item For round $i=1,2,\dots,k$:
    \begin{enumerate}
        \item \red{Obtain a query $q_i$ (not restricted to coordinates in $I$)}
        \item \red{For $j\in [n]$ let $c_{i,j}=\Enc(x_i,p_i,q_i(j))$.}
        \item Define the query $f_i:[n]\times\{0,1\}^\lambda\rightarrow\{-1,0,1\}$ where $f_i(j,x)=\Dec(x,p_j,c_{i,j})$. The description of $f_i$ includes $\{c_{i,j}\}_{j\in[n]}$ and $\{p_j\}_{j\in[n]}$.
        \item Give $f_i$ to $\AAA_2$ and obtain an answer $a_i\in[-1,1]$.
        \item Output $a_i$.
    \end{enumerate}
\end{enumerate}
\end{algorithm}

Unlike \texttt{WrapA}, the modified mechanism \texttt{WrapB} does not guarantee that the resulting $\texttt{WrapB}{\circ}\AAA$ is natural. Nevertheless, as we next explain, 
for every computationally efficient mechanism $\AAA$ it still holds that $\texttt{WrapB}{\circ}\AAA$ fails in \texttt{AdaptiveGameTwo}. 
To see this, consider algorithm  \texttt{WrapC}, which is another variant of \texttt{WrapA} and \texttt{WrapB} that has only oracle access to encryptions, via an oracle $\oracle$. We think of \texttt{WrapC} as an adversary to the MILR scheme.

\begin{algorithm}\caption{$\texttt{WrapC}$}

{\bf Setting:} Algorithm \texttt{WrapC} has access to an algorithm $\AAA{=}(\AAA_1,\AAA_2)$ with space complexity $s$ (designed for \texttt{AdaptiveGameSpace}) and it interacts with an adversary $\BBB$ (designed for \texttt{AdaptiveGameTwo}). It also has access to an encryption oracle $\oracle$.

{\bf Input:} An element $z$ (supposedly computed by $\AAA_1$), and a collection of $n$ public parameters $p_1,\dots,p_n\in\{0,1\}^{\poly(\lambda)}$.

\begin{enumerate}

    \item Instantiate $\AAA_2$ with $z$.
    
    \item For round $i=1,2,\dots,k$:
    \begin{enumerate}
        \item Obtain a query $q_i$ (not restricted to a set of coordinates)
        \item \red{For $j\in [n]$ let $c_{i,j}=\oracle(i,q_i(j))$.}
        \item Define the query $f_i:[n]\times\{0,1\}^\lambda\rightarrow\{-1,0,1\}$ where $f_i(j,x)=\Dec(x,p_j,c_{i,j})$. The description of $f_i$ includes $\{c_{i,j}\}_{j\in[n]}$ and $\{p_j\}_{j\in[n]}$.
        \item Give $f_i$ to $\AAA_2$ and obtain an answer $a_i\in[-1,1]$.
        \item Output $a_i$.
    \end{enumerate}
\end{enumerate}
\end{algorithm}

Observe that if \texttt{WrapC} has access to the oracle $\oracle_1$ then it behaves identically to \texttt{WrapA}, and if it has access to the oracle $\oracle_0$ then it behaves identically to \texttt{WrapB}. Formally,

\begin{align*}
&\Pr\left[
\texttt{AdaptiveGameTwo}\left( (\texttt{WrapB}{\circ}\AAA),\BBB,t,k \right)
=1\right]\\ 
&\quad=\Pr_{\substack{\vec{x},\vec{p},\BBB,\AAA_2,\Enc\\
z\leftarrow \AAA_1(\vec{x})\\
J\leftarrow J(\AAA_1,\vec{x},z,\vec{p})
}}\left[  
\texttt{AdaptiveGameTwo}\left( (\texttt{WrapC}^{\oracle_0(\vec{x},\vec{p},J,\cdot,\cdot)}(z,\vec{p}){\circ}\AAA),\BBB,t,k \right)
=1
\right]\\
&\quad\geq
\Pr_{\substack{\vec{x},\vec{p},\BBB,\AAA_2,\Enc\\
z\leftarrow \AAA_1(\vec{x})\\
J\leftarrow J(\AAA_1,\vec{x},z,\vec{p})
}}\left[  
\texttt{AdaptiveGameTwo}\left( (\texttt{WrapC}^{\oracle_1(\vec{x},\vec{p},J,\cdot,\cdot)}(z,\vec{p}){\circ}\AAA),\BBB,t,k \right)
=1
\right]-\negl(\Gamma(\lambda))\\
&\quad=\Pr\left[
\texttt{AdaptiveGameTwo}\left( (\texttt{WrapA}{\circ}\AAA),\BBB,t,k \right)
=1\right]-\negl(\Gamma(\lambda))\\
&\quad\geq\frac{3}{4}-\negl(\Gamma(\lambda))\geq\frac{2}{3}.
\end{align*}

Finally, observe that instead of composing \texttt{WrapB} with $\AAA$ (and thinking about them as a single mechanism) we could compose it instead with the adversary $\BBB$ to obtain a composed adversary that plays in \texttt{AdaptiveGameSpace}. That is,
$$
\texttt{AdaptiveGameTwo}\left( (\texttt{WrapB}{\circ}\AAA),\BBB,t,k \right)
\equiv
\texttt{AdaptiveGameSpace}\left( \AAA,(\texttt{WrapB}{\circ}\BBB),s,k \right),
$$

and hence,
\begin{align*}
&\Pr\left[\texttt{AdaptiveGameSpace}\left( \AAA,(\texttt{WrapB}{\circ}\BBB),s,k \right)=1\right]\\
&\qquad=
\Pr\left[
\texttt{AdaptiveGameTwo}\left( (\texttt{WrapB}{\circ}\AAA),\BBB,t,k \right)
=1\right]\\
&\qquad\geq\frac{2}{3}.
\end{align*}
\end{proof}

Theorem~\ref{thm:ADA} now follows by instantiating Theorem~\ref{thm:mainADAbody} with our MILR scheme, as specified in Theorem~\ref{thm:multi_first_statement}.


\section{Missing Proofs from Section~\ref{sec:constructME}}
\label{app:constructME}

\begin{proof}[Proof of Theorem~\ref{thm:multi1}]
We prove the result for a uniform adversary. The proof for a non-uniform adversary is similar.
Let $\BBB$ be a $\poly(\Gamma(\lambda))$ adversary for $\Pi$. Denote
$$
\adv_{\BBB}(\lambda) =
\left|
\Pr_{\vec{x},\vec{g},\BBB,\Enc}\left[  \BBB^{\oracle_0(\vec{x},\vec{g},[n],\cdot,\cdot)}(\vec{g})=1  \right]
-
\Pr_{\vec{x},\vec{g},\BBB,\Enc}\left[  \BBB^{\oracle_1(\vec{x},\vec{g},[n],\cdot,\cdot)}(\vec{g})=1  \right]
\right|.
$$
Our goal is to prove that $\adv_{\BBB}(\lambda) \leq \negl(\Gamma(\lambda))$.

\paragraph*{Reduction of the security of $\Pi$ from (multi-instance) security of $\Pi'$.}
We will use $\BBB$ to construct an adversary $\BBB_1$ for $n$ instances of $\Pi'$ with independent keys. We define the corresponding oracles for $\Pi'$.
\begin{enumerate}
	\item $\oracle_1^{(1)}(\vec{x'},[n],\cdot,\cdot)$ takes as input an index of a key $j\in[n]$ and a message $m$, and returns $\Enc'(x'_j,m)$.
	\item $\oracle_0^{(1)}(\vec{x'},[n],\cdot,\cdot)$ takes the same inputs and returns $\Enc'(x'_j,0)$.
\end{enumerate}
The advantage of $\BBB_1$ is defined as
    $$
\adv_{\BBB_1}(\lambda) =
\left|
\Pr_{\vec{x'},\BBB_1,\Enc'}\left[  \BBB_1^{\oracle_0^{(1)}(\vec{x'},[n],\cdot,\cdot)}()=1  \right]
-
\Pr_{\vec{x'},\BBB_1,\Enc'}\left[  \BBB_1^{\oracle_1^{(1)}(\vec{x'},[n],\cdot,\cdot)}()=1  \right]
\right|.
    $$

\begin{algorithm}\caption{$\BBB_1^{\oracle_{(\cdot)}^{(1)}(\vec{x'},[n],\cdot,\cdot)}()$}
\label{alg:b1}

{\bf Setting:} $\BBB_1$ runs against $n$ instances of $\Pi'$ and has access to $\BBB^{\oracle_{(\cdot)}(\vec{x},\vec{g},[n],\cdot,\cdot)}(\vec{g})$, which runs against $\Pi$.

\begin{enumerate}
  \item $\BBB_1$ selects $\vec{g} \leftarrow_R \mathcal{G}^n$ and gives it to $\BBB$ as input.
  \item  $\BBB_1$ forwards each query $(j,m)$ of $\BBB$ to its oracle and forwards the answers back to $\BBB$.
  \item $\BBB_1$ outputs the same output as $\BBB$.
\end{enumerate}

\end{algorithm}
The adversary $\BBB_1$ is given in Algorithm~\ref{alg:b1} and we analyze its runtime and advantage.

For this purpose, we first consider the distribution of the hashed keys of $\Pi$, given $\vec{g}$. Each key $x \in \{0,1\}^{\lambda}$ is selected uniformly and independently, and is hashed to $\lambda' = \lambda/10$ bits using a hash function selected uniformly from a universal hash function family. A standard application of the leftover hash lemma~\cite{HastadILL99} shows that the distribution of each key is $2^{-\Omega(\lambda)}$-close to the uniform distribution (in terms of statistical distance). Since the keys are selected independently and $n < 2^{-o(\lambda)}$, applying the triangle inequality (summing the statistical distance over all $n$ keys) shows that their joint distribution is $2^{-\Omega(\lambda)}$-close to uniform (given $\vec{g}$).

Since the keys $\vec{x'} \in (\{0,1\}^{\lambda'})^n$ of $\BBB_1$ are uniformly distributed independently of $\vec{g}$,
the simulation of $\BBB$ incurs an additive loss of $2^{-\Omega(\lambda)}$, namely $$\adv_{\BBB_1}(\lambda) \geq \adv_{\BBB}(\lambda) - 2^{-\Omega(\lambda)}.$$
The runtime of $\BBB_1$ is bounded by $\poly(\Gamma(\lambda))$ (as the runtime of $\BBB$).

\paragraph*{Reduction of (multi-instance) security of $\Pi'$ from single-instance security.}
Next, we construct an adversary $\BBB_2$ against a single instances of $\Pi'$ (with $n=1$), using the multi-instance adversary $\BBB_1$. The advantage of $\BBB_2$ is defined as
$$
\adv_{\BBB_2}(\lambda) =
\left|
\Pr_{\vec{x'},\BBB_2,\Enc'}\left[  \BBB_2^{\oracle_0^{(1)}(\vec{x'},[1],\cdot,\cdot)}()=1  \right]
-
\Pr_{\vec{x'},\BBB_2,\Enc'}\left[  \BBB_2^{\oracle_1^{(1)}(\vec{x'},[1],\cdot,\cdot)}()=1  \right]
\right|.
$$

\begin{algorithm}\caption{$\BBB_2^{\oracle_{(\cdot)}^{(1)}(\vec{x'},[1],\cdot,\cdot)}()$}
\label{alg:b2}

{\bf Setting:} $\BBB_2$ runs against a single instance of $\Pi'$ and has access to $\BBB_1^{\oracle_{(\cdot)}^{(1)}(\vec{x'},[n],\cdot,\cdot)}()$, which runs against $n$ instances of $\Pi'$.

\begin{enumerate}
  \item $\BBB_2$ selects an index $i \in [n]$ uniformly at random, and selects $n-1$ keys $x'_1,\ldots,x'_{i-1},x'_{i+1},\ldots,x'_{n}$, each chosen uniformly at random from $\{0,1\}^{\lambda'}$.
  \item $\BBB_2$ answers each query $(j,m)$ of $\BBB_1$ as follows:
  \begin{itemize}
    \item If $i = j$, $\BBB_2$ forwards query $(1,m)$ to its oracle and returns the answer to $\BBB_1$.
    \item If $j < i$, $\BBB_2$ returns $\Enc'(x'_j,m)$ to $\BBB_1$.
    \item If $j > i$, $\BBB_2$ returns $\Enc'(x'_j,0)$ to $\BBB_1$.
  \end{itemize}
  \item $\BBB_2$ outputs the same output as $\BBB_1$.
\end{enumerate}

\end{algorithm}
Adversary $\BBB_2$ is given in Algorithm~\ref{alg:b2}. In order to relate $\adv_{\BBB_2}(\lambda)$ to $\adv_{\BBB_1}(\lambda)$, we use a hybrid argument. Specifically, for $i \in \{0,\ldots,n\}$ define $\oracle^{(i)}$ as follows.
\begin{itemize}
	\item $\oracle^{(i)}(\vec{x'},[n],\cdot,\cdot)$ takes as input an index of a key $j\in[n]$ and a message $m$. If $j \leq i$, returns $\Enc'(x'_j,m)$. Otherwise, $j > i$ and it returns $\Enc'(x'_j,0)$.
\end{itemize}
Note that $\oracle^{(0)} = \oracle_0^{(1)}$  and $\oracle^{(n)} = \oracle_1^{(1)}$.
For $i \in [n]$, define
$$
\adv_{\BBB_1}^{(i)}(\lambda) =
\left|
\Pr_{\vec{x'},\BBB_1,\Enc'}\left[  \BBB_1^{\oracle^{(i-1)}(\vec{x'},[n],\cdot,\cdot)}()=1  \right]
-
\Pr_{\vec{x'},\BBB_1,\Enc'}\left[  \BBB_1^{\oracle^{(i)}(\vec{x'},[n],\cdot,\cdot)}()=1  \right]
\right|.
$$
We have
\begin{align*}
\adv_{\BBB_1}(\lambda) =
\sum_{i = 1}^{n}
\left(
\Pr_{\vec{x'},\BBB_1,\Enc'}\left[  \BBB_1^{\oracle^{(i-1)}(\vec{x'},[n],\cdot,\cdot)}()=1  \right]
-
\Pr_{\vec{x'},\BBB_1,\Enc'}\left[  \BBB_1^{\oracle^{(i)}(\vec{x'},[n],\cdot,\cdot)}()=1  \right]
\right) \leq \\
\sum_{i = 1}^{n}
\left|
\Pr_{\vec{x'},\BBB_1,\Enc'}\left[  \BBB_1^{\oracle^{(i-1)}(\vec{x'},[n],\cdot,\cdot)}()=1  \right]
-
\Pr_{\vec{x'},\BBB_1,\Enc'}\left[  \BBB_1^{\oracle^{(i)}(\vec{x'},[n],\cdot,\cdot)}()=1  \right]
\right| =
\sum_{i = 1}^{n} \adv_{\BBB_1}^{(i)}(\lambda).
\end{align*}
Depending on the oracle of $\BBB_2$, $\BBB_1$ interacts with either $\oracle^{(i-1)}$ or $\oracle^{(i)}$ and hence the advantage of $\BBB_2$ given that $i$ is selected is
$\adv_{\BBB_1}^{(i)}(\lambda)$. Since $i \in [n]$ is uniform,
\begin{align*}
\adv_{\BBB_2}(\lambda) =
\frac{1}{n} \sum_{i = 1}^{n} \adv_{\BBB_1}^{(i)}(\lambda) \geq
\frac{1}{n} \adv_{\BBB_1}(\lambda).
 \end{align*}
The runtime of $\BBB_2$ is bounded by $n \cdot \poly(\Gamma(\lambda)) = \poly(\Gamma(\lambda))$.
Hence, the security assumption about $\Pi'$ implies that $\adv_{\BBB_2}(\lambda) \leq \negl(\Gamma(\lambda))$. We therefore conclude that $\adv_{\BBB_1}(\lambda)$ and hence $\adv_{\BBB}(\lambda)$ is negligible in $\Gamma(\lambda)$.

\end{proof}

\section{Missing Proofs from Section~\ref{sec:preprocessing}}
\label{app:preprocessing}

The following result is taken form~\cite[Claim 2]{CorettiDGS18} (its original variant is due to~\cite{GoosLM0Z15,KothariMR17}).
\begin{lemma}
\label{lem:fixing}
Fix $z \in \{0,1\}^s$ and let $s_z = \lambda \cdot n - H_{\infty}(X_z)$ be the min-entropy deficiency of $X_z$. Let $\gamma > 0$ be arbitrary.
Then, for every $\delta > 0$, $X_z$ is $\gamma$-close to a convex combination of finitely many $(k,1 -\delta)$-dense sources for
$$k = \frac{s_z + \log 1/\gamma}{\delta \cdot \lambda}.$$
\end{lemma}

We now prove the following lemma.
\begin{lemma}
\label{lem:zleakage}
Fix $z \in \{0,1\}^s$ and let $s_z = \lambda \cdot n - H_{\infty}(X_z)$.

Let $\mathcal{G}$ be a family of universal hash functions with domain $\{0,1\}^\lambda$ and range $\{0,1\}^{\lambda'}$ and let $G \leftarrow_R (\mathcal{G})^n$.
Let $\delta >0, \gamma > 0$ be parameters and assume that $(1 - \delta) \lambda > \lambda' + \log n + 1$.

Then, there exists a family of convex combinations $V^{*}_{G,z} = \{V^{*}_{\vec{g},z}\}_{\vec{g} \in (\mathcal{G})^n}$ of $k$-bit-fixing $(n,2^{\lambda'})$-sources
where $$k = \frac{s_z + \log 1/\gamma}{\delta \cdot \lambda}$$ such that
$$\Delta[(G,G(X_z)),(G,V^{*}_{G,z})] \leq \sqrt{2^{-(1 - \delta) \lambda + \lambda' + \log n}} + \gamma.$$
\end{lemma}

Lemma~\ref{lem:zleakage} is based on lemmas~\ref{lem:leftover} and~\ref{lem:fixing}.

\begin{proof}
Fix $z \in \{0,1\}^s$. Let $s_z$, $\gamma > 0, \delta > 0$ be as in the lemma.

Then, by Lemma~\ref{lem:fixing} (applied with parameters $s_z,\gamma,\delta$), $\Delta[X_z,X^{*}_z] \leq \gamma$, where $X^{*}_z$ is a convex combination of finitely many $(k,1 - \delta)$-dense $(n,2^\lambda)$-sources for
$k = \frac{s_z + \log 1/\gamma}{\delta \cdot \lambda}$.

Let $X^{*}$ be a $(k,1 - \delta)$-dense source in the convex combination $X^{*}_z$.
Let $V^{*}$ be the corresponding $k$-bit-fixing source ($X^{*}$ and $V^{*}$ are fixed on the same indices to the same values).

Let $J \subseteq [n]$ denote the set of indices on which $X^{*}$ is not fixed and further denote $t \coloneqq |J| \geq n - k$.
Then, $X^{*}_J$ is a $(t,2^\lambda)$-source, which is $(1 - \delta)$-dense.

Since we assume that $(1 - \delta) \lambda > \lambda' + \log n + 1 > \lambda' + \log t + 1$,
we can apply Lemma~\ref{lem:leftover} to $X^{*}_J$ (with $\delta$ as above).
We conclude that
$$\Delta[(G,G(X^{*}_J)),(G,U)] \leq \sqrt{2^{-(1 - \delta) \lambda + \lambda' + \log t}} < \sqrt{2^{-(1 - \delta)\lambda + \lambda' + \log n}},$$
where $U \leftarrow_R (\{0,1\}^{\lambda'})^t$.

Since the remaining indices of $\overline{J}$ in $X^{*}$ and $V^{*}$ are fixed to the same values,
$$\Delta[(G,G(X^{*})),(G,V^{*}_{G})] \leq \sqrt{2^{-(1 - \delta)\lambda + \lambda' + \log n}},$$
where $V^{*}_{G} \in (\{0,1\}^{\lambda'})^n$ is a family of $|\overline{J}|$-bit-fixing sources, obtained from each $\vec{g} \in \mathcal{G}^n$ and $V^{*}$ by setting the indices of $\overline{J}$ to their values in $\vec{g}(V^{*})$, while the indices of $J$ are uniform in $(\{0,1\}^{\lambda'})^t$.

Let $(G,V^{*}_{G,z})$ be obtained by replacing every $(G,G(X^{*}))$ in $(G,G(X^{*}_z))$ with the corresponding $(G,V^{*}_{G})$.
Then, the above implies that
$$\Delta[(G,G(X^{*}_z)),(G,V^{*}_{G,z})] \leq \sqrt{2^{-(1 - \delta)\lambda + \lambda' + \log n}}.$$

Therefore,
\begin{align*}
&\Delta[(G,G(X_z)),(G,V^{*}_{G,z})] \leq \\
&\Delta[(G,G(X_z)) , (G,G(X^{*}_z)) ] + \Delta[(G,G(X^{*}_z)) ,(G,V^{*}_{G,z})] \leq
\gamma + \sqrt{2^{-(1 - \delta)\lambda + \lambda' + \log n}},
\end{align*}
where we have used the fact that
$$\Delta[(G,G(X_z)), (G,G(X^{*}_z)) ] \leq \Delta[X_z,X^{*}_z] \leq \gamma,$$
since $G$ is independent of $X_z$ and $X^{*}_z$.

\end{proof}

\paragraph{Proof of Lemma~\ref{lem:leakage}.}

Lemma~\ref{lem:leakage} is based on~Lemma~\ref{lem:zleakage} and the proposition below
which upper bounds the probability over $X$ and $F$ that $H_{\infty}(X_z)$ is low.
\begin{proposition}
\label{prop:min}
Let $s' > s$. Then,
$\Pr_{z \leftarrow F(X)}[H_{\infty}(X_z) < n \cdot \lambda - s'] < 2^{s - s'}$.
\end{proposition}

\begin{proof}
First, since a randomized function is a convex combination of deterministic ones,
we assume without loss of generality that $F$ is deterministic.
Fix $z \in \{0,1\}^s$ such that $H_{\infty}(X_z) < n \cdot \lambda - s'$.
Since $X \in (\{0,1\}^\lambda)^n$ is uniform, then $\Pr_X[F(X) = z] < 2^{-s'}$.

Let $K = \{z \in \{0,1\}^s \mid H_{\infty}(X_z) < n \cdot \lambda - s' \}$.
Then,
\begin{align*}
\Pr_{z \leftarrow F(X)}[H_{\infty}(X_z) < n \cdot \lambda - s'] = 
\sum_{z \in K} \Pr_X[F(X) = z] <
\sum_{z \in K} 2^{-s'} \leq 2^{s} \cdot 2^{-s'} = 2^{s - s'}.
\end{align*}

\end{proof}

\begin{proof}[Proof of Lemma~\ref{lem:leakage}]
Let $\delta >0, \gamma > 0, s' > s$ be parameters.

Consider $z \in \{0,1\}^s$ with $s_z \leq s'$ (i.e., $H_{\infty}(X_z) \geq n \cdot \lambda - s'$).
Since $(1 - \delta) \lambda > \lambda' + \log n + 1$, we invoke Lemma~\ref{lem:zleakage} and deduce that
there is a family of convex combinations $V^{*}_{G,z}$ of finitely many $k$-bit-fixing such that
\begin{align} \label{eq:1}
\Delta[(G,G(X_z)),(G,G(V^{*}_{G,z}))]  \leq \sqrt{2^{-(1 - \delta)\lambda + \lambda' + \log n}} + \gamma,
\end{align}
and
$$k = \frac{s' + \log 1/\gamma}{\delta \cdot \lambda}.$$

Define the family
$$V_{G,Z} = \{V_{G,z}\}_{z \in \{0,1\}^s} = \{V_{\vec{g},z}\}_{\vec{g} \in (\mathcal{G})^n, z \in \{0,1\}^s}$$
as follows: for each $z$
such that $s_z \leq s'$, set $V_{G,z} \coloneqq V^{*}_{G,z}$, while for each $z$
such that $s_z > s'$, set $V_{G,z}$ to be the uniform distribution over $(\{0,1\}^{\lambda'})^n$.
Then, $V_{G,Z}$ is a family of convex combinations $V_{\vec{g},z}$ of $k$-bit-fixing $(n,2^{\lambda'})$-sources,
for $k = \frac{s' + \log 1/\gamma}{\delta \cdot \lambda}$. To complete the proof, it remains to bound
$\Delta[(G,Z,G(X)),(G,Z,V_{G,Z})]$.

Denote by $\mathcal{E}$ the event that $s_z < s'$.
Then,
\begin{align*}
&\Delta[(G,Z,G(X)),(G,Z,V_{G,Z})] =
\E_{z \leftarrow F(X)}\Delta[(G,G(X_z)), (G,V_{G,z})] \leq \\
&\Pr[\mathcal{E}] \cdot 1 + \E_{z \leftarrow F(X){\mid \neg \mathcal{E}}} \Delta[(G,G(X_z)), (G,V_{G,z})] \leq \\
&2^{s - s'} + \sqrt{2^{-(1 - \delta)\lambda + \lambda' + \log n}} + \gamma,
\end{align*}
where the final inequality is based on Proposition~\ref{prop:min} and~(\ref{eq:1}).

\end{proof}

\section{Application to Communication Complexity}
\label{sec:app:cc}
In the two-player one-way communication game, inputs $A$ and $B$ are given to Alice and Bob, respectively, and the goal is for Alice to send a minimal amount of information to Bob, so that Bob can compute $f(A,B)$ for some predetermined function $B$. 
The communication cost of a protocol $\Pi$ is the size of the largest message in bits sent from Alice across all possible inputs $X$ and the (randomized) communication complexity is the minimum communication cost of a protocol that succeeds with probability at least $\frac{2}{3}$. 
In the distributional setting, $X$ and $Y$ are further drawn from some known underlying distribution. 

In our setting, suppose Alice has $n$ independent and uniform numbers $a_1,\ldots,a_n$ so that either $a_i\in GF(p)$ for all $i\in[n]$ or $a_i\in GF(2^t)$ for sufficiently large $t$ for all $i\in[n]$ and suppose Bob has $n$ independent and uniform numbers $b_1,\ldots,b_n$ from the same field, either $GF(p)$ or $GF(2^t)$. 
Then for any function $f(\langle a_1,b_1\rangle,\ldots,\langle a_n,b_n\rangle)$, where the dot products are taken over $GF(2)$ or $f(a_1\cdot b_1,\ldots,a_n\cdot b_mn$, where the products are taken over $GF(p)$, has the property that the randomized one-way communication complexity of computing $f$ with probability $\sigma$ is the same as the number of samples from $a_1,\ldots,a_n$ that Alice needs to send Bob to compute $f$ with probability $\sigma-\eps$. 
It is easy to prove sampling lower bounds for many of these problems, sum as $\sum_i a_i\cdot b_i\pmod{p}$ or $\MAJ(\langle a_1, b_1\rangle,\ldots,\langle a_n, b_n\rangle)$, and this immediately translates into communication complexity lower bounds by the above. 
The main intuition for these lower bounds is that the numbers $b_1,\ldots,b_n$ can be viewed as the hash functions that Bob has. 

In the context of Corollary~\ref{cor:leakage}, the input $A$ to Alice is the input $X=(X_1,\ldots,X_n)$ to the hash functions and the input $B$ to Bob is the set of hash functions $G=(G_1,\ldots,G_n)$. 
In particular, we have $X_i\in\{0,1\}^\lambda$ and $G_i\in\{0,1\}^t$ for each $i\in[n]$. 
Then informally speaking, Corollary~\ref{cor:leakage} shows that if a sampling protocol requires $\Omega(\eta^2 n)$ samples to achieve probability $\frac{1}{2}+\frac{\eta}{2}$ of success, i.e., $\frac{\eta}{2}$ advantage, then any algorithm that achieves probability $\frac{1}{2}+\eta$ of success, i.e., $\eta$ advantage, must use $\Omega(\eta^2 n\lambda)$ communication. 

\begin{theorem}
Let $X=(X_1,\ldots,X_n)$ with $X_i\in\{0,1\}^\lambda$ for all $i\in[n]$ be given to Alice and let $G=(G_1,\ldots,G_n)$ with $G_i\in\{0,1\}^t$ for all $i\in[n]$ be given to Bob. 
Suppose that for a function $f:(\{0,1\}^\lambda)^n\times(\{0,1\}^\lambda)^n\to\{0,1\}$, any protocol $\Pi$ that computes $f(X,G)$ with probability at least $\frac{1}{2}+\frac{\eta}{2}$ requires $k$ samples $(X_i,G_i)$ for some quantity $k$ that depends on $\eta$. 
Then any one-way protocol that computes $f(X,G)$ with probability at least $\frac{1}{2}+\eta$ must use $\Omega(s)$ communication, where $s=\frac{k\lambda}{2}-2\log_2\frac{6}{\eta}$. 
\end{theorem}
\begin{proof}
We define $h$ to be the function so that 
\[f(X,G)=h(G_1(X_1),G_2(X_2),\ldots,G_n(X_n)).\]
Setting $t=\lambda'=1$, $\delta = 1/2$, $n < 2^{\lambda/3}$, $\gamma = \eta/6$, $s = \Theta(\eta^2 n\lambda)$, and $s' = s + \log_2 (6/\eta)$, in the formulation of Corollary~\ref{cor:leakage}, then for sufficiently large $\lambda$ so that $2^{-c\lambda}<\frac{\eta}{6}$ for a fixed constant $c>0$, we have
\[\Delta[P_{(G,Z,G(X))},P_{(G,Z,V)}]\le\sqrt{2^{-(1-\delta)\lambda+\lambda'+\log n}}+\gamma+2^{s-s'}
\leq \sqrt{2^{-\lambda/2 + 1 + \lambda/3}} + \eta/6 + \eta/6
\leq \eta/2,\]
where $V$ is a $((G,Z),k)$-convex-bit-fixing $(n,2^{\lambda'})$-source, and $k=\frac{2(s + 2\log_2(6/\eta))}{\lambda}$. 

Suppose there exists a protocol $\Pi$ that computes $f(G,X)$ with probability more than $\frac{1}{2}+\eta$ using $o(s)$ space. 
By Corollary~\ref{cor:leakage}, the success probability is the convex combination of the success probabilities for $((G,Z),k')$-convex-bit-fixing $(n,2^{\lambda'})$-sources with $k'<k$. 
Thus there exists a protocol $\Pi'$ that produces a message $z$ by Alice such that Bob's probability of success on $g$ is at least $\frac{1}{2}+\eta$. 
By the law of total probability, the success probability of Bob is the weighted success probability of Bob conditioned on each $k$'-bit-fixing source in the convex combination and so there exists a $k'$-bit-fixing source with success probability at least $\frac{1}{2}+\frac{\eta}{2}$. 
Hence, Bob’s posterior distribution on $g_1,\ldots,g_n$ is jointly uniform on all coordinates outside of a fixed set $S$ of at most $k'$ indices, up to a negligible statistical distance. 
Since each $g_i$ is universal, Bob's posterior distribution is the same as if Alice had sampled the coordinates in $S$ and sent the coordinates to Bob. 
However, since $|S|\le k'<k$, this contradicts the assumption that any protocol $\Pi$ that computes $f(X,G)$ with probability at least $\frac{1}{2}+\frac{\eta}{2}$ requires $k$ samples $(X_i,G_i)$ for some quantity $k$ that depends on $\eta$. 
\end{proof}

We can apply Corollary~\ref{cor:leakage} in this setting to any relation $f(X, G)$ which has the form
\[h(G_1(X_1), G_2(X_2), \ldots, G_n(X_n))\]
for some relation $h$, provided the $G_i$ are independent and drawn from a universal hash function family with range $\{0,1\}^{\lambda'}$. 
Indeed, by viewing $Z$ in that theorem as the message sent from Alice to Bob, Corollary~\ref{cor:leakage} shows that the distributions
$P_{(G,Z,G(X))}$ and $P_{(G,Z,V)}$ have small variation distance, where $V$ is a $((G, Z), k)$-convex-bit-fixing source. 
Consequently, Bob, who is given $G = g$ and $Z = z$ for some particular realizations $g$ and $z$, cannot distinguish $g(X)$ from $V$, where $V$ is a $((g, z),k)$-convex-bit-fixing source, which means that conditioned on $G = g$ and $Z = z$, $V$ is a convex combination of distributions of $k$-bit-fixing sources, i.e., distributions on $(\{0,1\}^{\lambda'})^n$ that are fixed on at most $k$ coordinates and uniform on the rest. 

Since the average success probability over $g$ is at least $1/2 + \eta$, this implies there exists a $g$ for which the success probability is at least $1/2 + \eta$, and we can use this particular $g = (g_1, \ldots, g_n)$ to derive our lower bound. Similarly, we can fix a message $z$ of Alice for which the success probability of Bob on $g$ is at least $1/2 + \eta$. Now, by the law of total probability, the success probability of Bob is the weighted success probability of Bob conditioned on each $k$-bit-fixing source in the convex combination defining $V$, and so if Bob's success probability is at least $1/2 + \eta$ on $(g,z)$, there exists a $k$-bit-fixing source with success probability at least $1/2 + \eta$. This in turn means that Bob's posterior distribution on the $g_i(X_i)$ is jointly uniform on all coordinates outside of a fixed set $S$ of indices, up to a negligible statistical distance. 
 
The above means, in particular, that outside of $S$, $g_i$ cannot be the zero function, as otherwise $g_i(X_i)$ could not be uniform. But for many universal families, if $g_i \neq 0$ is fixed, then $g_i(X_i)$ being uniform means $X_i$ is uniform. 
This holds, in particular if $g_i(X_i) = \langle g_i, X_i \rangle \mod 2$. 
In this case, Bob's posterior distribution is the same as that obtained from Alice just sampling the coordinates $i \in S$ and sending the $X_i$ values to Bob, and so if we prove a sampling lower bound for $1/2 + \eta/2$ success probability, we obtain an arbitrary one-round communication lower bound for $1/2 + \eta$ success probability, assuming the parameters are set so that the statistical distance is at most say, $\eta/2$.

As an example, suppose $h$ is the majority function and for each $i$, $g_i(X_i) =  \langle g_i, X_i \rangle \mod 2$ for bitstrings $g_i$ and $X_i$. It is well-known that the set of all such $g_i$ form a universal family. To obtain $1/2+\eta$ success probability under the uniform distribution on $(G, X)$, assuming $\lambda = \Omega(\log 1/\eta)$ so that $g_i = 0^\lambda$ is unlikely (and for any fixed $g_i \neq 0^\lambda$, when $X_i$ is uniform, $\langle g_i, X_i \rangle$ is uniform), it is standard that it is necessary and sufficient for a sampling protocol to sample $\Theta(\eta^2 n)$ coordinates $i$ and send them to Bob, and then Bob outputs the empirical majority. 
In Corollary~\ref{cor:leakage} we can set $\lambda' = 1$, $\delta = 1/2$, $n < 2^{\lambda/3}$, $\gamma = \eta/6$, $s = \Theta(\eta^2 n \lambda)$, and $s' = s + \log_2 (6/\eta)$ to conclude that for $\eta/6 > 2^{-c \lambda}$ for a small constant $c > 0$,
\[\Delta[P_{(G,Z,G(X))},P_{(G,Z,V)}]\le\sqrt{2^{-(1-\delta)\lambda+\lambda'+\log n}}+\gamma+2^{s-s'}
\leq \sqrt{2^{-\lambda/2 + 1 + \lambda/3}} + \eta/6 + \eta/6
\leq \eta/2,\]
where $V$ is a $((G,Z),k)$-convex-bit-fixing $(n,2^{\lambda'})$-source, and 
$k=\frac{2(s + 2\log_2(6/\eta))}{\lambda}$. 
One needs $k = \Omega(\eta^2 n \lambda)$ by the sampling lower bound to achieve $1/2 + \eta/2$ success probability, which implies $s = \Omega(\eta^2 n \lambda - 2 \log_2(6/\eta))$ as a $1$-way arbitrary communication lower bound, since the success probability of the $1$-way protocol is at most $1/2 + \eta/2 + \eta/2 = 1/2 + \eta$.

\bibliographystyle{splncs03}

\end{document}